\documentclass[10pt,fleqn]{article}
\usepackage{amsfonts}
\usepackage{amsmath}
\usepackage{color}
\usepackage{indentfirst,latexsym,bm}
\usepackage{amssymb}
\usepackage{amsmath}
\usepackage{graphics}
\usepackage{graphicx}

\numberwithin{equation}{section}

 \newcommand{\ep}{\varepsilon}
 \newcommand{\p}{\partial}
 \newcommand{\be}{\begin{eqnarray}}
 \newcommand{\ee}{\end{eqnarray}}
 \newcommand{\bee}{\begin{eqnarray*}}
 \newcommand{\eee}{\end{eqnarray*}}
 \newcommand{\ra}{\rightarrow}
\newtheorem{thm}{Theorem}[section]

\newtheorem{lem}[thm]{Lemma}

\newtheorem{rmk}{Remark}[section]

\newtheorem{theorem}{Theorem}[section]

\newtheorem{lemma}[theorem]{Lemma}

\newtheorem{proposition}[theorem]{Proposition}

\newtheorem{assumption}[theorem]{Assumption}
\newenvironment{proof}[1][Proof]{\textbf{#1.} }{\ \rule{0.5em}{0.5em}}
\begin{document}

\title{\LARGE\bf Dynkin Game of Convertible Bonds and Their Optimal
Strategy\footnote{\noindent We thank the Editor-in-Chief, Goong
Chen, an Associate Editor and two referees for their valuable
comments and helpful suggestions, which helped to significantly
improve the presentation of the paper. The work is supported by NNSF
of China (No. 11271143, 11371155, 11326199), University Special
Research Fund for Ph.D. Program (No.20124407110001).}}
\date{}
\author{
Huiwen Yan~\thanks{School of Mathematics and Statistics,
  Guangdong University of Finance \& Economics, Guangzhou 510320, China,
  \texttt{hwyan10@gmail.com}}$\quad$
Fahuai Yi~\thanks{School of Mathematical Science, South China Normal
 University, Guangzhou 510631, China,
 \texttt{fhyi@scnu.edu.cn}} $\quad$
Zhou Yang~\thanks{School of Mathematical Science, South China Normal
 University, Guangzhou 510631, China,
 \texttt{yangzhou@scnu.edu.cn}. The corresponding author.}$\quad$
Gechun Liang~\thanks{Department of Mathematics, King's College
London, London WC2R 2LS, U.K.,
 \texttt{gechun.liang@kcl.ac.uk} }$\quad$}
 \maketitle

%

\begin{abstract}
This paper studies the valuation and optimal strategy of convertible
bonds as a Dynkin game by using the reflected backward stochastic
differential equation method and the variational inequality method.
We first reduce such a Dynkin game to an optimal stopping time
problem with state constraint, and then in a Markovian setting, we
investigate the optimal strategy by analyzing the properties of the
corresponding free boundary, including its position, asymptotics,
monotonicity and regularity. We identify situations when call
precedes conversion, and vice versa. Moreover, we show that the
irregular payoff results in the possibly non-monotonic conversion
boundary. Surprisingly, the price of the convertible bond is not
necessarily monotonic in time: it may even increase when time
approaches maturity.
\\


\end{abstract}

%
%

\baselineskip 18 pt
\section{Introduction}

Convertible bonds are often advertised as products with upside
potential and limited downside risk, since a convertible bond is
often supplemented with an option to exchange this bond for a given
number of shares. The bondholder decides whether to keep the bond,
in order to collect coupons, or to convert it to the firm's stocks.
She will choose a conversion strategy to maximize the bond value. On
the other hand, the issuing firm has the right to call the bond, and
presumably acts to maximize the equity value of the firm by
minimizing the bond value. This creates a \emph{two-person, zero-sum
Dynkin game}. One of the central questions for convertible bonds is
to study such a Dynkin game, and more importantly, the corresponding
optimal call and conversion strategies.

The study of convertible bonds dates back to Brennan and Schwartz
\cite{Brennan} and Ingersoll \cite{Ingersoll}. However, both the
call and conversion strategies are predetermined in these papers, so
neither of them need to address the free boundary problem that
arises if early conversion or early call is optimal. Sirbu et al
\cite{SPS} is one of the first to analyze the optimal strategy of
perpetual convertible bonds (see also Sirbu and Shreve \cite{SS1}
for the finite horizon counterpart). They reduce the problem from a
Dynkin game to an optimal stopping time problem, and discuss when
call precedes conversion,  and vice versa. Several more realistic
features of convertible bonds have been taken into account since
then. For example, Bielecki et al \cite{BCJR} consider the problem
of the decomposition of a convertible bond into bond component and
option component. Cr\'epey and Rahal \cite{Crepey} study the
convertible bond with call protection, which is typically path
dependent, and more recently, Chen et al \cite{Chen} consider the
tax benefit and bankruptcy cost for convertible bonds. For a
complete literature review, we refer to the aforementioned papers
with references therein.

In this paper, we first study the Dynkin game of convertible bonds
by using the reflected backward stochastic differential equation
(\emph{reflected BSDE for short}) method. Instead of regarding the
convertible written on the stock value which is endogenously
determined as the difference between the firm value and the bond
value, we take a reduced form approach by assuming that the firm's
stock value follows a general It\^o process exogenously.
Interestingly, similar to \cite{SPS} and \cite{SS1}, we can also
reduce the Dynkin game to an optimal stopping time problem with
state constraint, i.e. reducing the reflected BSDE with two
obstacles to a reflected BSDE with one obstacle and state
constraint. An important consequence of this representation result
is to allow us to identify when call precedes conversion, and vice
versa, which is in line with \cite{SS1}. That is, we show in
Propositions \ref{proposition1}-\ref{proposition2} that when the
coupon rate is bounded above by the interest rate times the
surrender price, the bondholder will always convert her bond first;
when the coupon rate is bounded below by the dividend rate times the
surrender price, the firm will always call the bond first; when the
coupon rate lies between the above two bounds, both the bondholder
and the firm will terminate the contract simultaneously. We show
that the above representation result holds in a general It\^o
process setting which is not necessarily Markovian, the latter of
which is the standing assumption in both \cite{SPS} and \cite{SS1}.

In the Markovian case, one way to study the optimal strategy of
convertible bonds is to analyze the properties of the free boundary
for the corresponding variational inequality (\emph{VI for short}).
Notwithstanding, the research on the free boundary analysis to
understand the optimal strategy of convertible bonds is rare
compared to the study on American options, for which the
corresponding free boundary has already been well studied. One of
the main reasons is that the corresponding Dynkin game (variational
inequality) is too complicated to study. By utilizing the
aforementioned representation result, we can reduce the
corresponding Dynkin game to an optimal stopping time problem with
state constraint, and this paves the way to study the properties of
the corresponding free boundary. The current authors have already
taken this path in some special cases (see \cite{Yan,Yang3,Yang2}).
For example, in \cite{Yang2} the authors assume that the issuer has
no right to call. In \cite{Yang3} the authors only consider the
surrender price and the final pre-specified price exactly equal, so
the corresponding free boundary is always monotonic. In \cite{Yan}
only the case that the coupon rate is less than the interest rate
times the surrender price is considered. In the present paper, we
attempt to close the previous gaps, and give a complete analysis of
the free boundary under different cases, including the position of
the free boundary with its asymptotics, monotonicity and regularity,
etc. In particular, we concentrate on the case with irregular payoff
(see Assumption 2.2).

There are several interesting properties of the free boundary as we
prove in Section 3. First, it is well known that the asymptotics of
the free boundary is more difficult to obtain than the asymptotics
of the solution to the equation, because the convergence of the
solution does not imply the convergence of the free boundary in
general, and it is very difficult to deduce the latter via partial
differential equation (\emph{PDE for short}) estimates. In Theorem
\ref{th5.3}, we manage to obtain the asymptotics of both solution
and free boundary. The main idea for the latter is as follows:  we
solve the corresponding
 perpetual problem, then use its solution to construct a sub-solution sequence
 and a super-solution sequence of the finite horizon problem, and show
the asymptotic behavior of the free boundary via the two sequences.

Secondly, the free boundary in the VI (\ref{eq2.10}) is
non-monotonic under some parameter assumptions (see Theorem
\ref{th5.4} and Figure 3.5). This is due to the singular terminal
payoff which results in the blowup of the time derivative of the
price near the maturity around the singular point. The
non-monotonicity of the free boundary results in the
non-monotonicity of the convertible bond price. In particular, the
price may go up near maturity. In order to prove such a
non-monotonicity property, we discuss its terminal asymptotic
behavior and its initial asymptotic behavior as time goes to
infinity, and prove that the terminal value is larger than the
initial value, but less than the value at some middle point.

Thirdly, a standard assumption to prove the smoothness of the free
boundary is that the difference between the solution and the lower
obstacle of the VI is increasing with respect to time (see
\cite{Fr2}). Without this monotonicity property, the regularity is
difficult to achieve as discussed in \cite{Bla, Pet}. Unless the
coupon rate is greater than the interest rate times the
pre-specified price for the final payoff as in Theorem \ref{th5.5},
this monotonicity condition does not hold, and the smoothness of the
free boundary is not obvious at all. In Theorem \ref{th5.6} we show
the smoothness of the free boundary even when this monotonicity
condition fails, by using a subtle coordinate transformation and the
comparison principle for VI.

 The rest of the paper is organized as follows: In Section 2, we
 formulate our pricing model of convertible bonds as a Dynkin game by using
 the reflected BSDE method. In Section 3, we study the optimal
 strategies of convertible bonds by analyzing the properties of the corresponding free boundary. Some technical details about
 the solvability of the VI are presented in the appendix.

\section{The Dynkin Game of Convertible Bonds}

In this section, we formulate the pricing problem of convertible
bonds as a zero-sum Dynkin game by using the reflected BSDE method.
Our main result in this section is to show that such a Dynkin game
can be reduced to an optimal stopping time problem with state
constraint.

For a fixed time horizon $T>0$, let $W$ be a one dimensional
Brownian motion on a filtered probability space
$(\Omega,\mathcal{F}, \mathbb{F}=\{{\mathcal{F}}_t\} ,\mathbb{P})$
satisfying the \emph{usual conditions}, where $\mathbb{F}$ is the
augmented filtration generated by the Brownian motion $W$, and
$\mathbb{P}$ is interpreted as the risk-neutral probability measure.
Consider a firm who issues convertible bonds with the coupon rate
$c$ and the maturity $T$. The convertible bond is written on the
firm's underlying stocks $S$, whose price process under the
risk-neutral probability measure $\mathbb{P}$ is given by \be
\label{eq2.1}
 S_s=S_t+\int_t^s (r_u-q_u) S_u \,du +\int_t^s \sigma_u S_u dW_u,
\ee for $0\leq t\leq s\leq T$,
where $r,\,q,\,\sigma$ represent the
risk-free interest rate, the dividend rate and the volatility,
 respectively.

\begin{assumption}\label{Assumption1}
The coupon rate $c$, the risk-free interest rate $r$, the dividend
rate $q$ and the volatility $\sigma$ are {$\mathbb{F}$-progressively
measurable} and
uniformly bounded. Additionally, the volatility is positive
$\sigma_t>0$, $a.s.$ for $t\in[0,T]$.
\end{assumption}

Consider an investor purchasing a share of convertible bond from the
issuer at any stating time $t\in[0,T]$. Assume there is no default
for the firm. By holding the convertible bond, she will continuously
receive the coupon rate $c$ from the issuer until the contract is
terminated. Prior to the contract maturity $T$, the investor has the
right to convert her bond to the firm's stocks, while the firm has
the right to call the bond and force the bondholder to surrender her
bond to the firm. Hence there are three situations that the contract
will be terminated: (1) if the firm calls the bond at some
$\mathbb{F}$-stopping time $\tau$ first, the bondholder will receive
a pre-specified surrender price $K$ at time $\tau$; (2) if the
investor chooses to convert her bond  at some $\mathbb{F}$-stopping
time $\theta$ first or both players choose to stop the contract
simultaneously, the bondholder will obtain $\gamma S_{\theta}$ at
time $\theta$ from converting her bond with a pre-specified
conversion rate $\gamma\in(0,+\infty)$; (3) if neither players take
any action during the contract period, then at the maturity $T$, the
investor must sell her bond to the firm with a pre-specified price
$L$ or convert it to the firms' stocks with the conversion rate
$\gamma$, so she will obtain $\max\{L, \gamma S_T\}$. In summary,
the investor will obtain the following discounted payoff at the
starting time $t\in[0,T]$: \be\nonumber
 P(\tau,\theta)&=&\int_t^{\tau\wedge\theta}\,R(t,u)\,c_u\,du
 +R(t,\tau)\,K\,I_{\{\tau<\theta\}}
 +R(t,\theta)\,\gamma\,S_\theta\,I_{\{\theta\leq\tau,\,\theta<T\}}
 \\[2mm]                                                                              \label{eq2.2}
 &&+\,R(t,T)\,\max\{L,\,\gamma S_T\}\,I_{\{\tau\wedge\theta=
 T\}},
 \ee
where $\tau,\theta\in\mathcal{U}_{t,T}$, the set of all
$\mathbb{F}-$stopping times taking values in $[\,t,T\,]$, and
$R(t,u)=\exp\{-\int_t^u r_sds\}$ is the discount rate from $t$ to
$u$ in the risk-neutral world.

The investor will choose $\theta\in\mathcal{U}_{t,T}$ to maximize
$P(\tau,\theta)$, while the firm will choose
$\tau\in\mathcal{U}_{t,T}$ to minimize $P(\tau,\theta)$. Hence we
have the upper value and lower value, respectively,
\begin{align*}
 \overline{V}_t{=}
 \mathop{{\rm ess.inf}}_{\tau\in{\cal U}_{t,T}}
 \mathop{{\rm ess.sup}}_{\theta\in{\cal U}_{t,T}}
 \mathbb{E}\,[\,P(\tau,\theta)|{\cal F}_t\,];\\
\underline{V}_t{=}
 \mathop{{\rm ess.sup}}_{\theta\in{\cal U}_{t,T}}
 \mathop{{\rm ess.inf}}_{\tau\in{\cal U}_{t,T}}
 \mathbb{E}\,[\,P(\tau,\theta)|{\cal F}_t \,]
\end{align*}
of a corresponding Dynkin game (see \cite{Karatzas} for the
definition of Dynkin game). The value of this game exists if there
exists some process $V$ such that
$$V_t=\overline{V}_t=\underline{V}_t,\ \ \ a.s.\ \text{for}\  t\in[0,T],$$
and $V_t$ is the time $t$ value of this convertible bond by
no-arbitrage principle (see Chapter 36 of \cite{WIL}). It is
standard to show that if there exits a Nash equilibrium point
$(\tau^*_t,\theta^*_t)
 \in {\cal U}\,_{t,\,T}\times{\cal U}\,_{t,\,T}$ such that
$$
 \mathbb{E}\,[\,P(\tau^*_t,\theta)|{\cal F}_t  \,]\leq
 \mathbb{E}\,[\,P(\tau^*_t,\theta^*_t)|{\cal F}_t \,]\leq
 \mathbb{E}\,[\,P(\tau,\theta^*_t)|{\cal F}_t \,],\;\; a.s.\ \
 \text{for}\;\tau,\theta\in {\cal U}\,_{t,\,T},\qquad
 $$
then the value of this game $V$ exists and is given by
 \be                                                                                  \label{eq2.4}
 V_t=\mathbb{E}\,[\,P(\tau^*_t,\theta^*_t)|{\cal F}_t \,].
 \ee
The Nash equilibrium point $(\tau^*_t,\theta^*_t)$ is called the
optimal strategy for such a convertible bond, where $\tau^*_t$ and
$\theta^*_t$ represent the optimal calling and conversion strategy,
respectively. The conversion payoff $\gamma S$ is usually called the
lower obstacle, and the surrender price $K$ called the upper
obstacle.

\begin{assumption}
The risk-free interest rate is no less than the dividend rate:
$r_t\geq q_t\geq0$,
 $a.s.$ for $t\in[0,T]$,
 and the surrender price
is greater than the maturity payment: $K>L>0$.
\end{assumption}

The first assumption $r_t\geq q_t\geq0$ is natural. If $K<L$, then
the pre-specified price $L$ is irrelevant, since the firm could
always call with the surrender price $K$ before the maturity to
avoid paying more (see \cite{SS1}). If $K=L$, as shown in
\cite{Yang2}, the terminal value in the effective domain (state
constraint) of the corresponding VI is just constant $K$, so the
problem is relatively standard to study, and the corresponding free
boundary is always monotonic. In this paper, we mainly consider the
case $K>L$ which results in the singular terminal value across the
free boundary, and this makes the problem much more complicated and
involved. Moreover, the free boundary is not necessarily monotonic
in this case.

In the following, we represent the optimal strategy
$(\tau^*_t,\theta^*_t)$ and the price $V_t$ of the convertible bond
in terms of the solution of reflected BSDE. Note that it is not
always true that the conversion payoff (the lower obstacle) $\gamma
S$ is dominated by the surrender price (the upper obstacle) $K$, so
we have to resort to a reflected BSDE with the state constraint as
follows.

\begin{lemma}
Let $(Y,Z, K^+,K^-)$ be the  unique solution of the following
reflected BSDE on $[t,\sigma_t^*]$:
 \be\label{reflectedBSDE}
 \left\{
 \begin{array}{l}
 {\displaystyle Y_s=\max\{\gamma S_T,L\}I_{\{\sigma_t^*=T\}}+\gamma S_{\sigma_t^*}I_{\{\sigma_t^*<T\}}
 +\int_s^{\sigma_t^*}(c_u-r_uY_u)du-\int_s^{\sigma_t^*}Z_udW_u}
 \vspace{2mm}\\
 \qquad +\displaystyle\int_s^{\sigma_t^*}dK_u^+-\int_s^{\sigma_t^*}dK_u^-,\ \ \ \gamma S_s\leq Y_s\leq
 K,\ \ \ \mbox{for}\;s\in[\,t,\sigma_t^*\,],
 \vspace{2mm} \\
 {\displaystyle
 \int_t^{\sigma_t^*}(Y_u-\gamma
 S_u)dK^+_u=\int_t^{\sigma_t^*}(K-Y_u)dK^-_u=0},
 \end{array}
 \right.
 \ee
where
$$\sigma_t^*=\inf\{u\geq t: S_u\geq K/\gamma\}\wedge T.$$
Then the value of the convertible bond is given by $V_t=Y_t$ and the
optimal strategy is given by
$$
 \tau^*_t=\inf\{s\geq t:Y_s=K\}\wedge\sigma_t^*,\qquad
 \theta^*_t=\inf\{s\geq t:Y_s=\gamma S_s\}\wedge\sigma_t^*.
$$
\end{lemma}

The proofs of the above representation result and the well posedness
of (\ref{reflectedBSDE}) are similar to Theorem 4.1 of Cvianic and
Karatzas \cite{Karatzas} with the fixed maturity $T$ replaced by the
random maturity $\sigma_t^*$, so we omit the proofs and refer to
\cite{Karatzas} for the details. Our main result in this section is
to reduce (\ref{reflectedBSDE}) into an optimal stopping time
problem with state constraint.

First note that if $S_t\geq K/\gamma$, i.e. the lower obstacle is
greater than the upper obstacle, then $\sigma_t^*=t$, and in this
case, both the investor and the firm will choose to terminate the
contract at the same time $\tau_t^*=\theta_t^*=t$, and the value of
the convertible bond is nothing but $V_t=\gamma S_t$. Hence, in the
following we only consider the case $S_t< K/\gamma$.

\begin{proposition}\label{proposition1} Suppose that
{$c_s\leq r_s K$ a.s. on $s\in[t,\sigma_t^*]$. }
Then the value of the convertible bond is given by $V_t=Y_t^1$,
where $Y^1$ solves the following reflected BSDE:
\be\label{reflectedBSDE1}
 \left\{
 \begin{array}{l}
 {\displaystyle Y_s^1=\max\{\gamma S_T,L\}I_{\{\sigma_t^*=T\}}+\gamma S_{\sigma_t^*}I_{\{\sigma_t^*<T\}}
 +\int_s^{\sigma_t^*}(c_u-r_uY_u^1)du-\int_s^{\sigma_t^*}Z_u^1dW_u}
 \vspace{2mm}\\
 \qquad +\displaystyle\int_s^{\sigma_t^*}dK_u^{1,+},\ \ \ Y_s^1\geq\gamma S_s,\ \ \ \mbox{for}\;s\in[\,t,\sigma_t^*\,],
 \vspace{2mm} \\
 {\displaystyle\int_t^{\sigma_t^*}(Y_u^1-\gamma
 S_u)dK^{1,+}_u=0}.
 \end{array}
 \right.
 \ee
 {In particular, if $c_s<r_sK$ a.s. on $s\in[t,\sigma_t^*]$, then
$Y_s^1<K$ on $s\in[t,\sigma_t^*)$}, so the optimal strategy is given
by
$$
 \tau^*_t=\sigma_t^*,\qquad
 \theta^*_t=\inf\{s\geq t:Y_s^1=\gamma S_s\}\wedge\sigma_t^*.
$$
\end{proposition}

\begin{proof}
 {We first prove that $Y^1_s\leq K$ on $s\in[t,\sigma_t^*]$.}
Then $(Y^1,Z^1,K^{1,+},0)$ is the solution to (\ref{reflectedBSDE}).
Indeed, consider the following auxiliary reflected BSDE:
\begin{equation*}
 \left\{
 \begin{array}{l}
 {\displaystyle \bar{Y}_s^1=K
 +\int_s^{\sigma_t^*}(r_uK-r_u\bar{Y}_u^1)du-\int_s^{\sigma_t^*}\bar{Z}_u^1dW_u}
 \displaystyle+\int_s^{\sigma_t^*}d\bar{K}_u^{1,+},\ \ \ \bar{Y}_s^1\geq\gamma S_s,\ \ \
\mbox{for}\;s\in[\,t,\sigma_t^*\,],
 \vspace{2mm} \\
 {\displaystyle\int_t^{\sigma_t^*}(\bar{Y}_u^1-\gamma
 S_u)d\bar{K}^{1,+}_u=0},
 \end{array}
 \right.
 \end{equation*}
which obviously has a unique solution
 $(\bar{Y}^1_s,\bar{Z}_s^1,\bar{K}_s^{1,+})=(K,0,0)$.
Since $$K\geq \max\{\gamma S_T,L\}I_{\{\sigma_t^*=T\}}+\gamma
S_{\sigma_t^*}I_{\{\sigma_t^*<T\}},$$ and
 {$r_sK-r_s Y_s^1\geq c_s-r_sY_s^1$ a.s. }on
$s\in[t,\sigma_t^*]$, the comparison principle of reflected BSDE
(see \cite{ElKaroui}) implies that $Y_s^1\leq \bar{Y}_s^1=K$ on
$s\in[t,\sigma_t^*]$.

Next we show that $Y_s^1<K$ on $s\in[t,\sigma_t^*)$
 {if $c_s<r_sK$ a.s. on $s\in[t,\sigma_t^*]$}
 . If not, there exits $\bar{s}\in[t,\sigma_t^*)$ such that
$Y_{\bar{s}}^1=K$. Note that we must have $Y^1_{\bar{s}}>\gamma
S_{\bar{s}}$ (otherwise {$\gamma S_{\bar{s}}\geq Y_{\bar{s}}^1=K$
would imply}
$\bar{s}=\sigma_t^*$). Define
$$\theta_{\bar{s}}^*=\inf\{s\geq {\bar{s}}:Y_s^1=\gamma S_s\}\wedge\sigma_t^*.$$
Then $Y_{\theta_{\bar{s}}^*}^1=\gamma S_{\theta_{\bar{s}}^*}\leq K$.
Since {$Y_s>\gamma S_s$ and $d\bar{K}_s^{1,+}=0$ on
$[\bar{s},\theta_{\bar{s}}^*)$,}
(\ref{reflectedBSDE1}) reads
$$Y_{\bar{s}}^1=Y_{\theta_{\bar{s}}^*}^1
 +\int_{\bar{s}}^{\theta_{\bar{s}}^*}(c_u-r_uY_u^1)du-\int_{\bar{s}}^{\theta_{\bar{s}}^*}Z_u^1dW_u.$$
Consider the following auxiliary BSDE:
$$\hat{Y}_{\bar{s}}^1=K
 +\int_{\bar{s}}^{\theta_{\bar{s}}^*}(r_uK-r_u\hat{Y}_u^1)du-\int_{\bar{s}}^{\theta_{\bar{s}}^*}\hat{Z}_u^1dW_u,$$
 which obviously has a unique solution
 $(\hat{Y}_s^1,\hat{Z}_s^1)=(K,0)$. Then the strict comparison
 principle of BSDE (see \cite{ElKaroui1997}) implies that $Y_{\bar{s}}^1<K$.
\end{proof}

From the above proposition, if $c_s < r_sK$, the value of the
convertible bond $V_t$ is strictly less than the surrender price $K$
before the termination of the contact, so the firm will not call the
bond back until the contract is terminated at $\sigma_t^*$ , and the
investor will always convert her bond first.

\begin{proposition}\label{proposition3} Suppose that $c_s\geq q_s K$ a.s. on $s\in[t,\sigma_t^*]$. Then the value of the convertible bond is
given by $V_t=Y_t^2$, where $Y^2$ solves the following reflected
BSDE: \be\label{reflectedBSDE3}
 \left\{
 \begin{array}{l}
 {\displaystyle Y_s^2=\max\{\gamma S_T,L\}I_{\{\sigma_t^*=T\}}+\gamma S_{\sigma_t^*}I_{\{\sigma_t^*<T\}}
 +\int_s^{\sigma_t^*}(c_u-r_uY_u^2)du-\int_s^{\sigma_t^*}Z_u^2dW_u}
 \vspace{2mm}\\
 \qquad -\displaystyle \int_s^{\sigma_t^*}dK_u^{2,-},\ \ \ Y_s^2\leq K,\ \ \ \mbox{for}\;s\in[\,t,\sigma_t^*\,],
 \vspace{2mm}\\
 {\displaystyle
 \int_t^{\sigma_t^*}(K-Y_u^2)dK^{2,-}_u=0}.
 \end{array}
 \right.
 \ee
 In particular, if $c_s> q_sK$ a.s. on $s\in[t,\sigma_t^*]$, then
$Y_s^2>\gamma S_s$ on $s\in[t,\sigma_t^*)$, so the optimal strategy
is given by
$$
 \tau^*_t=\inf\{s\geq t:Y_s^2=K\}\wedge\sigma_t^*,\qquad
 \theta^*_t=\sigma_t^*.
 $$
\end{proposition}

\begin{proof} We first prove that $Y_s^2\geq \gamma S_s$ on
$s\in[t,\sigma_t^*]$. Then $(Y^2,Z^2,0,K^{2,-})$ is the solution to
(\ref{reflectedBSDE}). Indeed, consider the following auxiliary
reflected BSDE:
\begin{equation*}
 \left\{
 \begin{array}{l}
 {\displaystyle \bar{Y}_s^2=\gamma S_{\sigma_t^*}
 +\int_s^{\sigma_t^*}(\gamma q_uS_u-r_u\bar{Y}_u^2)du-\int_s^{\sigma_t^*}\bar{Z}_u^2dW_u}
 \displaystyle-\int_s^{\sigma_t^*}d\bar{K}_u^{2,-},\ \ \ \bar{Y}_s^2\leq K,\ \ \
\mbox{for}\;s\in[\,t,\sigma_t^*\,],
 \vspace{2mm} \\
 {\displaystyle\int_t^{\sigma_t^*}(K-\bar{Y}_u^2)d\bar{K}^{2,-}_u=0},
 \end{array}
 \right.
 \end{equation*}
which obviously has a unique solution
 $(\bar{Y}^2_s,\bar{Z}_s^2,\bar{K}_s^{2,-})=(\gamma S_s,\gamma
 \sigma_s{S_s},0)$. Since
 $$\gamma S_{\sigma_t^*}\leq \max\{\gamma S_T,L\}I_{\{\sigma_t^*=T\}}+\gamma
S_{\sigma_t^*}I_{\{\sigma_t^*<T\}},$$ and
 $
 \gamma q_sS_s-r_sY_s^2\leq  q_s K-r_s Y_s^2 \leq c_s-r_sY_s^2
 $ on $s\in[t,\sigma_t^*]$,
the comparison principle implies that $Y_s^2\geq \bar{Y}_s^2=\gamma
S_s$ on $s\in[t,\sigma_t^*]$.

Next we show that $Y_s^2>\gamma S_{s}$ on $s\in[t,\sigma_t^*)$ if
$c_s> q_sK$ a.s. on $s\in[t,\sigma_t^*]$. If not, there exits
$\bar{s}\in[t,\sigma_t^*)$ such that $Y_{\bar{s}}^2=\gamma
S_{\bar{s}}$. Note that we must have $Y^2_{\bar{s}}<K$ (otherwise
 {$\gamma S_{\bar{s}}= Y_{\bar{s}}^2\geq K$}
 would imply that $\bar{s}=\sigma_t^*$). Define
$$\tau_{\bar{s}}^*=\inf\{s\geq {\bar{s}}:Y_s^2=K\}\wedge\sigma_t^*.$$
Then $Y_{\tau_{\bar{s}}^*}^2=K\geq \gamma S_{\tau_{\bar{s}}^*}$.
Since $Y_s^2<K$, and $d\bar{K}_s^{2,-}=0$ on
$[\bar{s},\tau_{\bar{s}}^*)$, (\ref{reflectedBSDE3}) reads
$$Y_{\bar{s}}^2=Y_{\tau_{\bar{s}}^*}^2
 +\int_{\bar{s}}^{\tau_{\bar{s}}^*}(c_u-r_uY_u^2)du-\int_{\bar{s}}^{\tau_{\bar{s}}^*}Z_u^2dW_u.$$
Consider the following auxiliary BSDE:
$$\hat{Y}_{\bar{s}}^2=\gamma S_{\tau_{\bar{s}}^*}
 +\int_{\bar{s}}^{\tau_{\bar{s}}^*}(\gamma q_uS_u-r_u\hat{Y}_u^2)du-\int_{\bar{s}}^{\tau_{\bar{s}}^*}\hat{Z}_u^2dW_u,$$
 which obviously has a unique solution
 $(\hat{Y}_s^2,\hat{Z}_s^2)=(\gamma S_s,\gamma \sigma_s{S_s})$. Then the strict comparison
 principle implies that $Y_{\bar{s}}^1>\gamma S_{\bar{s}}$.
\end{proof}

From the above proposition, if $c_s>q_sK$, the value of the
convertible bond $V_t$ is strictly larger than the converting value
$\gamma S_t$ before the termination of the contact, so the investor
will not convert her bond until the contract is terminated at
$\sigma_t^*$ , and the firm will always call the bond first.

 {By repeating the arguments as in the proofs of Propositions
 \ref{proposition1} and \ref{proposition3}, we obtain that
 the price can be represented as the solution of the
 following BSDE~\eqref{reflectedBSDE2} if $q_s K\leq c_s\leq r_sK$. In particular, }
if $q_s K<c_s< r_sK$, then the value $V_s$ of the convertible bond
is bounded between $(\gamma S_s, K)$ before the termination of the
contact. Hence, neither the investor will convert her bond nor the
firm will call the bond back until the contract is terminated at
$\sigma_t^*$.


\begin{proposition}\label{proposition2} Suppose that
 {$q_s K\leq c_s\leq r_sK$ a.s. }on
$s\in[t,\sigma_t^*]$. Then the value of the convertible bond is
given by $V_t=Y_t^3$, where $Y^3$ solves the following BSDE on
$[t,\sigma_t^*]$:
\begin{equation}\label{reflectedBSDE2}
Y_s^3=\max\{\gamma S_T,L\}I_{\{\sigma_t^*=T\}}+\gamma
S_{\sigma_t^*}I_{\{\sigma_t^*<T\}}
 +\int_s^{\sigma_t^*}(c_u-r_uY_u^3)du-\int_s^{\sigma_t^*}Z_u^3dW_u.
\end{equation}
In particular, if $q_s K< c_s< r_sK$ a.s. on $s\in[t,\sigma_t^*)$,
then $Y_s^3\in(\gamma S_s,K)$ on $s\in[t,\sigma_t^*)$, so the
optimal strategy is given by
$$
 \tau^*_t=\theta_t^*=\sigma_t^*.
$$

\end{proposition}

\section{The Optimal Strategy of Convertible Bonds}

In this section, we further consider the optimal strategy of
convertible bonds in the Markovian case by investigating the
properties of the corresponding calling/conversion boundaries.

\begin{assumption}
Assume that all the coefficients are constants:
$c_t=c,r_t=r>0,q_t=q,$ and  $\sigma_t=\sigma$ for $t\in[0,T]$.
\end{assumption}

Due to the above Markovian assumption, we know that there exists a
function $V(S,t)$ such that $V_t=V(S_t,t)$. Define the following
domains
\begin{align*}
\text{Conversion domain \bf CV}=&\  \{(S,t)\in(0,\infty)\times[0,T):V(S,t)=\gamma S\};\\
\text{Calling domain \bf CL}=&\  \{(S,t)\in(0,\infty)\times[0,T):V(S,t)=K\neq\gamma S\};\\
\text{Continuation domain {\bf CT}}=&\
\{(S,t)\in(0,\infty)\times[0,T):\gamma S<V(S,t)<K\}.
\end{align*}
The intersecting line between the conversion domain ${\bf CV}$ and
the continuation domain ${\bf CT}$ is called the conversion boundary
$C(t)$, while the intersecting line between the calling domain ${\bf
CL}$ and the continuation domain ${\bf CT}$ is called the calling
boundary $H(t)$.

From the Feynman-Kac formula for the solution of reflected BSDE and
the viscosity solution of VI (see Section 8 of \cite{ElKaroui}),
Proposition \ref{proposition1} implies that if {$c\leq q K\ (\leq
rK)$£¬}
 then $V(S,t)=V^1(S,t)$ where $V^1$ solves the
following VI with the state constraint: \be \label{eq2.10}
 \left\{
 \begin{array}{l}
 \p_t V^1+{\cal L}_0 V^1=-c,
 \hspace{2cm}\mbox{if}\;\;V^1>\gamma S\;\;\mbox{and}\;\;
 (S,\,t)\in D_T,
 \vspace{2mm} \\
 \p_t V^1+{\cal L}_0 V^1\leq -c,
 \hspace{2cm}\mbox{if}\;\;V^1=\gamma S\;\;\mbox{and}\;\;
 (S,\,t)\in D_T,
 \vspace{2mm} \\
 V^1\left({K/\gamma},t\right)\;=K,
 \hspace{3.9cm}0\leq t\leq T,
 \vspace{2mm} \\
 V^1(S,T)\;=\max\{L,\gamma S\},\hspace{2.3cm}
 0\leq S\leq {K/\gamma},
 \end{array}
 \right.
 \ee
where $${\cal L}_0 V^1={\frac{\sigma^2}2}
 S^2\,\p_{SS} V^1 +\left(\,r-q
 \,\right)S\,\p_{S} V^1-rV^1, \quad D_T{=}(0,K/\gamma)\times [\,0,T).$$
 Herein, $D_T$ is the effective domain
(the state constraint) of our problem, since in the domain
$[K/\gamma,\infty)\times[\,0,T)$, $V(S,t)=\gamma S$, so the investor
will always choose to convert, and
$[K/\gamma,\infty)\times[0,T)\subset\bf CV$. Moreover, if $c<qK$,
Proposition \ref{proposition1} also implies that $V^1(t,S)=Y^1_t<K$
on $(t,S)\in D_T$. Hence, we have
${\bf CL}\cap D_T=\O$.

Similarly, Proposition \ref{proposition3} implies that if
 {$c\geq rK\ (\geq qK)$,}
 then $V(S,t)=V^2(S,t)$ where $V^2$ solves the
following VI with the state constraint: \be \label{eq2.101}
 \left\{
 \begin{array}{l}
 \p_t V^2+{\cal L}_0 V^2=-c,
 \hspace{2cm}\mbox{if}\;\;V^2<K\;\;\mbox{and}\;\;
 (S,\,t)\in D_T,
 \vspace{2mm} \\
 \p_t V^2+{\cal L}_0 V^2\geq -c,
 \hspace{2cm}\mbox{if}\;\;V^2=K\;\;\mbox{and}\;\;
 (S,\,t)\in D_T,
 \vspace{2mm} \\
 V^2\left({K/\gamma},t\right)\;=K,
 \hspace{3.9cm}0\leq t\leq T,
 \vspace{2mm} \\
 V^2(S,T)\;=\max\{L,\gamma S\},\hspace{2.3cm}
 0\leq S\leq {K/\gamma}.
 \end{array}
 \right.
 \ee
Moreover, if $c>rK$, Proposition \ref{proposition3} also implies
that
 {$V^2(t,S)=Y^2_t>\gamma S$ on $(t,S)\in D_T$, so ${\bf CV}\cap D_T=\O$.}

Finally, if
 {$qK\leq c\leq rK$,}
 then $V(S,t)=V^3(S,t)$ where $V^3$ solves
the following Dirichlet problem: \be \label{eq2.5}
 \left\{
 \begin{array}{l}
 \p_t V^3+{\cal L}_0 V^3=-c,
 \hspace{3.2cm} \mbox{in}\; D_T,\;
 \vspace{2mm} \\
 V^3\left({K/\gamma},t\right)\;=K,
 \hspace{3.4cm}0\leq t\leq T,
 \vspace{2mm} \\
 V^3(S,T)\;=\max\{L,\gamma S\},\hspace{2cm}
 0\leq S\leq {K/\gamma}.
 \end{array}
 \right.
 \ee
Moreover, the strong maximum principle (see \cite{Evans}) implies
that $V^3(t,S)=Y_t^3\in(\gamma S,K)$ on $(t,S)\in D_T$ (not only for
the case $qK<c<rK$).

Therefore, the analysis of the calling/conversion strategies boils
down to the properties of the free boundaries imbedded in the above
three PDE problems.

The {VI} (\ref{eq2.101}) for the case $c>rK$ has been studied in
\cite{Yan}. In such a case, the bondholder will not convert in the
domain $(S,t)\in D_T$, and the calling boundary $H(t)$ is always
monotonic (See Figure $3.1$). The problem is therefore relatively
standard. The PDE (\ref{eq2.5}) for the case
 {$qK\leq c\leq rK$}
 is trivial in the sense that neither the bondholder will
convert nor the firm will call in the domain $(S,t)\in D_T$ (See
Figure $3.2$).
 {We leave the explicit solution of the PDE (\ref{eq2.5}) in Appendix B.}
In this paper, we mainly consider the VI (\ref{eq2.10}) for the case
$c< qK$. The situation in such a case is much more complicated and
involved (See Figure $3.3-3.5$). The conversion boundary $C(t)$ may
even lose the monotonicity property in such a case due to the
singular payoff.

\begin{picture}(280,120)(150,0)
 \put(180,10){\vector(1,0){120}}
 \put(180,6){\vector(0,1){105}}
 \put(180,100){\line(1,0){100}}
 \put(178,98){$\bullet$}\put(165,98){$T$}
 \textcolor[rgb]{0.00,0.00,1.00}{\thicklines{\qbezier(180,25)(220,30)(243,70)}}
 \put(240,68){$\bullet$}
 \put(178,23){$\bullet$}
 \put(180,70){$\gamma\,S<V<K$}\put(260,35){$V=\gamma\,S$}\put(205,12){$V=K$}
 \put(183,108){$t$}\put(305,10){$S$}
 \put(177,7){$\bullet$}\put(165,8){$O$}
 \put(241,7){$\bullet$}\put(241,-5){$K/\gamma$}
 \textcolor[rgb]{1.00,0.00,0.00}{\thicklines\put(243,10){\line(0,1){60}}}
 \textcolor[rgb]{1.00,0.00,0.00}{\thicklines\put(239,70){\line(0,1){30}}}
 \put(200,82){${\bf CT}$}\put(260,50){${\bf CV}$}
 \put(215,25){${\bf CL}$}
 \put(195,40){${\bf H}(t)$} \put(230,87){${\bf C}(t)$}
\end{picture}
\begin{picture}(160,120)(210,0)
 \put(180,10){\vector(1,0){120}}
 \put(180,6){\vector(0,1){105}}
 \put(180,100){\line(1,0){100}}
 \put(180,30){$\gamma\,S<V<K$}\put(260,30){$V=\gamma\,S$}
 \put(178,98){$\bullet$}\put(165,98){$T$}
 \put(183,108){$t$}\put(305,10){$S$}
 \put(177,7){$\bullet$}\put(165,12){$O$}
 \put(241,7){$\bullet$}\put(241,-5){$K/\gamma$}
 \textcolor[rgb]{1.00,0.00,0.00}{\thicklines\put(243,10){\line(0,1){90}}}
 \put(200,45){${\bf CT}$}\put(270,45){${\bf CV}$}
 \put(235,75){${\bf C}(t)$}
\end{picture}
 \begin{center}
 Figure$\!$ 3.1.$\;\;\;c>rK$
 \hspace{4.2cm}$\;\;$
 Figure$\!$ 3.2.$\;\;\;qK\leq c\leq rK$
 \end{center}
 \begin{picture}(280,120)(135,0)
 \put(180,10){\vector(1,0){120}} \put(180,6){\vector(0,1){105}}
 \put(180,100){\line(1,0){100}}\put(178,98){$\bullet$}\put(165,98){$T$}
 \put(183,108){$t$}\put(305,10){$S$}
 \textcolor[rgb]{1.00,0.00,0.00}{\thicklines\qbezier(220,100)(260,80)(263,30)
 \put(263,10){\line(0,1){20}}}
 \put(177,7){$\bullet$}\put(182,14){$O$}
 \put(218,97){$\bullet$}
 \put(261,7){$\bullet$}\put(261,-5){$K/\gamma$}
 \put(210,70){${\bf CT}$}\put(270,70){${\bf CV}$}
 \put(185,55){$\gamma\,S<V<K$}\put(265,55){$V=\gamma\,S$}
 \put(254,20){${\bf C}(t)$}
 \end{picture}
 \begin{picture}(160,120)(200,0)
 \put(180,10){\vector(1,0){120}} \put(180,6){\vector(0,1){105}}
 \put(180,100){\line(1,0){100}}\put(178,98){$\bullet$}\put(165,98){$T$}
 \put(183,108){$t$}\put(305,10){$S$}
 \textcolor[rgb]{1.00,0.00,0.00}{\thicklines\qbezier(220,100)(245,60)(255,9)}
 \put(177,7){$\bullet$}\put(182,14){$O$}
 \put(218,97){$\bullet$}
 \put(265,7){$\bullet$}\put(265,-5){$K/\gamma$}
 \put(210,65){${\bf CT}$}\put(270,65){${\bf CV}$}
 \put(180,50){$\gamma\,S<V<K$}\put(260,50){$V=\gamma\,S$}
 \put(243,15){${\bf C}(t)$}
 \end{picture}
 \begin{center}
 $\;$Figure$\!$ 3.3.$\;\;\;rL\leq c<qK,\,c>{(\alpha_+-1)\,rK\over \alpha_+}$
 $\qquad\qquad$
 Figure$\!$ 3.4.$\;\;\;rL\leq c<qK,\,c\leq {(\alpha_+-1)rK\over \alpha_+}$
 \end{center}

 \begin{picture}(260,130)(30,0)
 \put(180,10){\vector(1,0){120}} \put(180,6){\vector(0,1){105}}
 \put(180,100){\line(1,0){100}}
 \put(178,98){$\bullet$}\put(165,98){$T$}
 \put(183,108){$t$}\put(305,10){$S$}
 \textcolor[rgb]{1.00,0.00,0.00}{\thicklines\qbezier(220,99)(255,85)(260,60)
 \qbezier(260,60)(255,45)(235,40)
 \qbezier(235,40)(206,35)(205,10)}
 \put(177,7){$\bullet$}\put(182,14){$O$}
 \put(219,96){$\bullet$}\put(210,104){$L/\gamma$}
 \put(202,7){$\bullet$}
 \put(268,7){$\bullet$}\put(270,-5){$K/\gamma$}
 \put(220,70){${\bf CT}$}\put(270,70){${\bf CV}$}
 \put(190,55){$\gamma\,S<V<K$}\put(265,55){$V=\gamma\,S$}
 \put(198,15){${\bf C}(t)$}
\end{picture}
 \begin{center}
 Figure$\!$ 3.5.$\;\;\;c\leq rL(\alpha_+-1)/\alpha_+$
 \end{center}

\subsection{Properties of the Conversion Boundary}

In this subsection, we prove the properties of the free boundary
$C(t)$ of (\ref{eq2.10}), such as its position, asymptotic property,
monotonicity property, regularity property etc.  We first show in
Theorem \ref{th4.1} that the solution $V^1$ is not only the
viscosity solution, but also the strong solution to (\ref{eq2.10}):
$V^1\in W^{2,\,1}_p(D_T)\cap C(\overline{D_T}\,)$ with $p>1$.

Since (\ref{eq2.10}) is degenerate, we first transform it into a
familiar non-degenerate VI via the following transformation:
 \be                                                                            \label{eq3.1}
 u(x,\tau)=V^1(S,\,t),\qquad \tau=T-t,\qquad x=\ln S-\ln K+\ln\gamma.
 \ee
Then it is not difficult to check that $u$ is governed by
\be
\label{eq4.1}
 \left\{
 \begin{array}{l}
 \p_\tau u-{\cal L} u=c
 \qquad\mbox{if}\;u>Ke^x\;\;\mbox{and}\;\;
 (x,\,\tau)\in \Omega_T,
 \vspace{2mm} \\
 \p_\tau u-{\cal L} u\geq c
 \qquad\mbox{if}\;u=Ke^x\;\;\mbox{and}\;\;
 (x,\,\tau)\in \Omega_T,
 \vspace{2mm} \\
 u(0,\,\tau)\;=K,\hspace{2.6cm}
 0\leq \tau\leq T,
 \vspace{2mm} \\
 u(x,0)\;=\max\{L,Ke^x\},\hspace{1cm}
 x\leq 0,
 \end{array}
 \right.
 \ee
 where
 \be                                                                                \label{eq3.3}
 {\cal L} u={\frac{\sigma^2}2}
 \,\p_{xx} u +\left(\,r-q-{\frac{\sigma ^2}2}
 \,\right)\,\p_x u-ru,\quad
 \Omega_T{=}(-\infty,0)\times(0,T\,].
 \ee

 \begin{thm}                                                                           \label{th4.1}
 For the case $c< qK$, the VI (\ref{eq4.1}) has a unique
 strong solution $u\in W_{p,\,loc}^{2,\,1}(\Omega_T)\cap C(\,\overline{\Omega_T}\,)$
 with $p>1$. Moreover, $\p_x u\in C(\Omega_T)$ and
 we have the following estimates:
 \be                                                                                \label{eq4.11}
 &\displaystyle{\max\left\{Ke^x,{c\over r}
 +{rL-c\over r}\,e^{-r\tau}\right\}
 \leq u\leq K}
 &\mbox{in}\;\;\overline{\Omega}_T,
 \\[2mm]                                                                             \label{eq4.12}
 &0\leq\p_x  u\leq K\,e^{x}
 &\mbox{in}\;\;\overline{\Omega}_T.
 \ee
 If furthermore $c\geq rL$ holds, we also have the following estimate:
 \begin{equation}                                                                                  \label{eq4.13}
 \p_\tau  u \geq 0\;\; a.e.\
 \mbox{in}\;\;\Omega_T.
 \end{equation}
\end{thm}

The proof of Theorem \ref{th4.1} is quite long and relatively
standard, so we leave its proof in the appendix.

\begin{picture}(160,130)(230,0)
 \put(270,10){\vector(1,0){140}} \put(360,6){\vector(0,1){110}}
 \put(358,118){$\tau$}\put(413,10){$x$}
 {\thicklines\qbezier(310,10)(343,30)(348,110)}
 \multiput(350,10)(0,5){21}{\line(0,1){2}}
 \put(308,8){$\bullet$}\put(280,-2){$c_{0}=\underline{X}$}
 \put(348,8){$\bullet$}\put(345,-2){$c_{\infty}$}
 \put(290,55){${\bf CT}_x$}\put(270,30){$u>K\,e^x$}
 \put(350,55){${\bf CV}_x$} \put(340,30){$u=K\,e^x$}
 \put(320,70){$c(\tau)$}
\end{picture}
\begin{picture}(160,130)(200,0)
 \put(270,10){\vector(1,0){140}} \put(360,6){\vector(0,1){110}}
 \put(358,118){$\tau$}\put(413,10){$x$}
 {\thicklines\qbezier(310,10)(350,30)(360,90)
 \put(360,90){\line(0,1){15}}}
 \put(308,8){$\bullet$}\put(293,-2){$c_{0}=\underline{X}$}
 \put(358,87){$\bullet$}\put(362,90){$\overline{T}$}
 \put(290,55){${\bf CT}_x$}\put(270,30){$u>K\,e^x$}
 \put(360,55){${\bf CV}_x$} \put(350,30){$u=K\,e^x$}
 \put(337,70){$c(\tau)$}
\end{picture}
 \begin{center}
 Figure$\!$ 3.6.$\;\;\;c\geq rL,\,c\leq rK(\alpha_+-1)/\alpha_+$
 \qquad
 Figure$\!$ 3.7.$\;\;\;c\geq rL,\,c>rK(\alpha_+-1)/\alpha_+$
 \end{center}

\begin{picture}(130,130)(130,0)
 \put(270,10){\vector(1,0){140}} \put(360,6){\vector(0,1){110}}
 \put(358,118){$\tau$}\put(413,10){$x$}
 {\thicklines\qbezier(310,10)(343,20)(345,30)
 \qbezier(345,30)(340,45)(300,50)
 \qbezier(300,50)(287,55)(283,110)}
 \multiput(280,10)(0,5){21}{\line(0,1){2}}
 \multiput(310,10)(0,5){4}{\line(0,1){2}}
 \put(308,6){$\bullet$}\put(298,-2){$c_{0}=\ln L-\ln K$}
 \put(278,6){$\bullet$}\put(275,0){$c_{\infty}$}
 \put(308,25){$\bullet$}\put(313,28){$\underline{t}$}
 \put(260,65){${\bf CT}_x$}\put(270,35){$u>K\,e^x$}
 \put(320,65){${\bf CV}_x$} \put(350,35){$u=K\,e^x$}
 \put(290,95){$c(\tau)$}
 \put(270,6){$\bullet$}\put(265,-2){$\underline{X}$}
\end{picture}
 \begin{center}
 Figure$\!$ 3.8.$\;\;\;c\leq rL(\alpha_+-1)/\alpha_+$
 \end{center}

We denote ${\bf CV}_x,\,{\bf CT}_x,\,c(\tau)$ as the counterparts of
${\bf CV},\,{\bf CT},\,C(t)$ by the transformation (\ref{eq3.1}),
respectively. {From~\eqref{eq4.12}, $u-Ke^x$ is decreasing with
respect to $x$, so the conversion boundary is given as
$c(\tau)=\inf\{x\leq0:u(x,\tau)=Ke^x\}$, and}
 the conversion region and the
continuation region can further be characterized as
\begin{align*}
{\bf CV}_x&=\{(x,\tau)\in(-\infty,0)\times(0,T]: u(x,\tau)=
Ke^x\}=\{(x,\tau)\in(-\infty,0)\times(0,T]: x\geq c(\tau)\};\\
{\bf CT}_x&=\{(x,\tau)\in(-\infty,0)\times(0,T]: u(x,\tau)>
Ke^x\}=\{(x,\tau)\in(-\infty,0)\times(0,T]: x< c(\tau)\}.
\end{align*}

Our main result in this section is to prove that the conversion and
continuation regions have the following shapes under different
parameter assumptions (see Figures $3.6-3.8$). Note that Figure 3.8
shows that the conversion boundary is non-monotonic, and the price
may go up when time approaches maturity around the starting point
$c_0$ of the conversion boundary.

In the following, we prove the position, the asymptotics, the
monotonicity and the regularity of the free boundary $c(\tau)$.

\begin{thm} ({Position of the free boundary})

                                                                         \label{th5.2}
For the case $c<qK$, the free boundary $c(\tau)$ of the variational
inequality (\ref{eq4.1}) has the following properties:
 \smallskip

 \noindent (1) ${\bf CV}_x\subset\{x\geq
 \underline{X}\,\}$, so that ${\bf
 CT}_x\supset\{x<\underline{X}\,\}$ and
 $c(\tau)\in[\,\underline{X}\,,0\,]$, where $\underline{X}\,\stackrel{\Delta}{=}\ln c-\ln K-\ln q$.
 \smallskip

 \noindent (2) There exists a positive constant $\underline{t}$ such
that
 $c(\tau)>c_0$ for any $\tau\in(0,\underline{t}\,]$ where $c_{0}\stackrel{\Delta}{=}\max\{\underline{X},\,\ln L-\ln
 K\}$.
 \smallskip

 \noindent (3) The starting point of $c(\tau)$ is
 $(c(0),0)$ with $c(0)\stackrel{\Delta}{=}\lim\limits_{\tau\rightarrow0^+}c(\tau)=c_{0}$.

\end{thm}

\begin{proof} (1).
 According to (\ref{eq4.1}), in the domain
 ${\bf CV}_x\cap\Omega_T,\,V=Ke^x$ and it must hold
 $$
 c\leq \p_\tau V-{\cal L}V=\p_\tau (K\,e^x)-{\cal L}(K\,e^x)=q\,K\,e^x
 \Rightarrow x\geq \underline{X}.
 $$
 Hence, ${\bf CV}_x\cap\Omega_T\subset\{x\geq \underline{X}\,\}$. Since
 ${\bf CL}_x\cap\Omega_T=\O$, then
 ${\bf CT}_x\supset\{x<\underline{X}\,\}$ and  $c(\tau)\geq\underline{X}\,$.

 (2). The proof is divided into two cases:

 \noindent \underline{Case 1}: If $c_0=\underline{X}$ (see Figures $3.6$ and $3.7$), then it is sufficient to prove that
 $c(\tau)>\underline{X}$ for any $\tau>0$. Suppose not. Property (1) implies that
 there exists a $t_1>0$ such that $c(t_1)=\underline{X}$.
We deduce that in
 the domain ${\cal
 N}=(-\infty,\underline{X}\,)\times(0,t_1\,]: \,{u>Ke^x}$, and $u$ satisfies
 $$
 \p_\tau u-{\cal L} u=c\geq qKe^x=\p_\tau(Ke^x)-{\cal L}(Ke^x),\quad
 u(\underline{X}\,,t_1)=Ke^x|_{x=\underline{X}\,}.
 $$
In view of the Hopf lemma (see \cite{Evans}), we obtain that $\p_x
(u-Ke^x)(\underline{X}\,,t_1)<0$.
 On the other hand, Theorem \ref{th4.1} implies that $\p_x u\in C(\Omega_T)$. It means that
 $\p_x u$ continuously crosses the free boundary $c(\tau)$, and
 $\p_x (u-Ke^x)(\underline{X}\,,t_1)=0$. Hence, we have a contradiction.

 \noindent \underline{Case 2}: If $c_0=\ln L-\ln K$ (see Figure $3.8$), then it is sufficient to prove that
 there exists a positive constant $\underline{t}$ such that
 $$
 c(\tau)>\ln L-\ln K,\qquad\forall\;\;\tau\in(0,\underline{t}\,].
 $$
What we need to prove is that there exists a positive constant
$\underline{t}$ such that
 \be                                                                             \label{eq5.2}
 u(\ln L-\ln K,\tau)>K\,e^x|_{x=\ln L-\ln K}=L,
 \qquad\forall\;\;\tau\in(0,\underline{t}\,].
 \ee
Indeed, we denote $w$ as the solution of PDE (\ref{eq3.2}), then
 $w$ takes the explicit form of (\ref{eq3.5}). Hence, the A-B-P maximum principle (see \cite{Tao}) implies that $u\geq w$ in
 $\overline{\Omega}_T$.

 In order to prove (\ref{eq5.2}), we use the
 explicit form of (\ref{eq3.5}) to estimate asymptotic behavior of
 $w(\ln L-\ln K,\tau)$ as $\tau\rightarrow0^+$. It is not difficult to
 check that as $\tau\rightarrow0^+$, we have
 \bee
 &&
 \left\{
 \begin{array}{ll}
 \mbox{if}\;x>0,\quad
 &\Phi_1(x,\tau,t)=\Phi_2(x,\tau,t)=1+o(\,\sqrt\tau\,),
 \vspace{2mm} \\
 \mbox{if}\;x<0,\quad
 &\Phi_1(x,\tau,t)=\Phi_2(x,\tau,t)=o(\,\sqrt\tau\,),
 \vspace{2mm} \\
 \mbox{if}\;x=0,\quad &\Phi_1(x,\tau,t)={1\over2}-{\sigma\,\alpha_1\over\sqrt{2\pi}}\,\sqrt{\tau-t}
 +o(\,\sqrt\tau\,),
 \vspace{2mm} \\
 \mbox{if}\;x=0,\quad &\Phi_2(x,\tau,t)={1\over2}-{\sigma\,(\alpha_1+1)\over\sqrt{2\pi}}\,\sqrt{\tau-t}
 +o(\,\sqrt\tau\,),
 \end{array}
 \right.\\[2mm]
 &&w(\ln L-\ln K,\tau)-L
 \\[2mm]
 &=&o(\,\sqrt{\tau}\,)
 +L\,\left(\,{1\over2}-{\sigma\,\alpha_1\over\sqrt{2\pi}}\,\sqrt{\tau}
 +o(\,\sqrt{\tau}\,)\,\right)
 -L\,\left(\,{1\over2}-{\sigma\,(\alpha_1+1)\over\sqrt{2\pi}}\,\sqrt{\tau}
 +o(\,\sqrt{\tau}\,)\,\right)
 \\[2mm]
 &=&{\sigma\,L\over\sqrt{2\pi}}\,\sqrt{\tau}
 +o(\,\sqrt{\tau}\,).
 \eee
 Hence, there exists a positive constant $\underline{t}$ satisfies (\ref{eq5.2}).

 (3). Since we have proved the property (2),
 it is sufficient to show
 $$
 \limsup\limits_{\tau\rightarrow0^+}c(\tau)\leq c_0.
 $$
 The above inequality is obvious if we can prove that for any fixed
 $x_1>c_0$, there exists a positive constant $\delta^*$ such that
 \be                                                                                      \label{eq5.3}
 u(x_1,\tau)=K\,e^x|_{x=x_1},\quad\forall\;\tau\in[\,0,\delta^*\,].
 \ee
 Indeed, for any fixed $x_1>c_0$, we construct a function such that
 $$
 W(x,\tau)=Ke^x+\delta(x-x_1)^2,
 \quad (x,\tau)\in{\cal N}
 \stackrel{\triangle}{=}[\,x_1-\delta,x_1+\delta\,]\times[\,0,\delta^*\,],
 $$
 where $\delta, \,\delta^*$ are positive constants to be determined. We first
 assume $\delta$ small enough so that $x_1-\delta>c_0$ and $x_1+\delta<0$.
 Next, we show that $u\leq W$ in $\overline{{\cal N}}$. Indeed, it is easy to check that in the domain ${\cal N},\,W$
satisfies
 \bee
 \p_\tau W-{\cal L}W
 &=&qKe^x-\delta\,[\sigma^2+(2r-2q-\sigma^2)(x-x_1)-r(x-x_1)^2]
 \\[2mm]
 &\geq&qKe^{x_1-\delta}-\delta\,[\sigma^2+(2r+2q+\sigma^2)\delta\,]
 >c-\delta\,[\sigma^2+(2r+2q+\sigma^2)\delta\,],\qquad
 \eee
 where we have used $x_1-\delta>c_0\geq \underline{X}$ in the last inequality.
 Choose $\delta$ small enough such that
 $$
 \p_\tau W-{\cal L}W>c\;\;\mbox{in}\;{\cal N}.
 $$
 Moreover, it is clear that
 $$
 W(x_1\pm\delta,0)>Ke^x\Big|_{(x_1\pm\delta,0)}=u(x_1\pm\delta,0).
 $$
 Recalling $u\in C(\,\overline{\Omega}_T\,)$, we deduce that there exists
 a positive constant $\delta^*$ such that
 $$
 W(x_1\pm\delta,\tau)\geq u(x_1\pm\delta,\tau),\quad
 \forall\;\;\tau\in[\,0,\delta^*\,].
 $$
 Hence, $W$ satisfies
 \bee
 \left\{
 \begin{array}{l}
 \p_\tau W-{\cal L} W\geq c,\quad
 W\geq Ke^x, \quad \mbox{in}\;{\cal N},
 \vspace{2mm} \\
 W\geq u,\qquad \mbox{on}\;\p_p {\cal N}.
 \end{array}
 \right.
 \eee
 The A-B-P maximum principle (see \cite{Tao}) implies that $u\leq W$ in the domain
 $\overline{{\cal N}}$. In particularly, $u\leq W=Ke^x$ on the line
 $x=x_1,\,\tau\in(0,\delta^*\,]$. By combining $u\geq Ke^x$,
 we obtain (\ref{eq5.3}).
\end{proof}

 Next, we analyze the asymptotic behavior of the free boundary and
 the solution of the VI (\ref{eq4.1}) as $\tau\rightarrow\infty$:

 \begin{thm} (Asymptotics of the free boundary)

                                                                          \label{th5.3}
 For the case $c<qK$, the free boundary $c(\tau)$ and the solution $u(x,\tau)$ of the VI (\ref{eq4.1})
 has the following asymptotic properties:\smallskip

 \noindent(1) If furthermore $c\leq rK(\alpha_+-1)/\alpha_+$ holds, where $\alpha_+$ is
 defined in Lemma \ref{lm4.1}, then we have (see Figure 3.6 and Figure 3.8)
 \bee
 &&\lim_{\tau\rightarrow+\infty}c(\tau)=c_{\infty}\stackrel{\Delta}{=}
 \ln \left(\,{\alpha_+\over \alpha_+-1}\;{c\over rK}\,\right),
 \\[2mm]
 &&\lim_{\tau\rightarrow+\infty}u(x,\tau)= u_{1,\,\infty}(x)
 \stackrel{\Delta}{=}
 \left\{
 \begin{array}{ll}
 {\displaystyle {K\over \alpha_+}\,\exp\Big\{\,\alpha_+x+(1-\alpha_+)\,c_{\infty}\,\Big\}
 +{c\over r}, }\;\; &x<c_{\infty},
 \vspace{2mm}\\
 Ke^x,  &c_{\infty}\leq x\leq 0.\qquad
 \end{array}
 \right.
 \eee
 \noindent(2) If furthermore $c>rK(\alpha_+-1)/\alpha_+$ holds, then there exists a positive
 constant $\overline{T}$ such that the free boundary $c(\tau)$ ends at the point
 $(0,\overline{T})$ (see Figure 3.7), i.e.,
 \bee
 &&
 c(\tau)=0,\;\;\text{for}\;\tau\in[\overline{T},T],\quad
 {\bf CT}_x\supset(-\infty,0)\times[\,\overline{T},T\,],
 \\[2mm]
 &&\lim_{\tau\rightarrow+\infty}u(x,\tau)= u_{2,\,\infty}(x)
 \stackrel{\Delta}{=}
 Ke^{\alpha_+x}+{c\over r}\,\left(\,1-e^{\alpha_+x}\,\right),
 \quad x\leq 0.
 \eee
\end{thm}
\begin{rmk}                                                                         \label{rm5.1}
 In fact, the above results imply that the solution $u(x,\tau)$ and
 the free boundary $c(\tau)$ of the finite horizon problem
 converge to the solution $u_{1,\,\infty}$ ( or $u_{2,\,\infty}$)
 and the free boundary $c_\infty$ (or $0$) of the corresponding
 perpetual problem as time tends to infinity, respectively.
\end{rmk}

\begin{proof}
 The proof is divided into five steps:

 \noindent \underline{Step 1}: Construct a super-solution and a sub-solution of
 the VI (\ref{eq4.1}).

 For any fixed $t>0$, we denote $\overline{u}_t$ as the
 $W^2_{p,\,loc}(\Omega)\cap C(\,\overline{\Omega}\,)$ solution of
 the following VI:
 \be                                                                              \label{eq5.4}
 \left\{
 \begin{array}{l}
 -{\cal L}\,\overline{u}_t=c+re^{-rt/2},
 \qquad\mbox{if}\;\overline{u}_t>Ke^x\;\;\mbox{and}\;\;
 x\in \Omega\triangleq(-\infty,0),
 \vspace{2mm} \\
 -{\cal L}\,\overline{u}_t\geq c+re^{-rt/2},
 \qquad\mbox{if}\;\overline{u}_t=Ke^x\;\;\mbox{and}\;\;
 x\in \Omega,
 \vspace{2mm} \\
 \overline{u}_t(0)\;=K.
 \end{array}
 \right.
 \ee
 We will give the explicit solution of the VI
 (\ref{eq5.4}) in Step 2. Denote
 $$
 W=\overline{u}_t+e^{rt/2-r\tau}-e^{-rt/2}.
 $$
We claim that $W$ is a super-solution of VI (\ref{eq4.1}) if $t$ is
large
 enough. In fact, it is not difficult to check that
 \bee
 \left\{
 \begin{array}{l}
 \p_\tau W-{\cal L}\,W\geq c
 \;\;\mbox{and}\;\;W\geq Ke^x,
 \vspace{2mm} \\
 W(0,\tau)\geq K=u(0,\tau),\qquad \qquad\quad
 0\leq \tau\leq t,
 \vspace{2mm} \\
 W(x,0)\geq K+e^{rt/2}-e^{-rt/2}\geq K\geq\max\{L,Ke^x\}=u(x,0),\quad
 x\leq 0.\;\;\mbox{(note that $r>0$)}
 \end{array}
 \right.
 \eee
By applying the comparison principle for VI {(see \cite{Yan1})}, we
deduce that
 \be                                                                             \label{eq5.5}
 u\leq W=\overline{u}_t+e^{rt/2-r\tau}-e^{-rt/2}.
 \ee

 Next, denote $\underline{u}\,_t$ as the
 $W^2_{p,\,loc}(\Omega)\cap C(\,\overline{\Omega}\,)$ solution of the
 following VI:
 \be                                                                              \label{eq5.6}
 \left\{
 \begin{array}{l}
 -{\cal L}\,\underline{u}\,_t=c-re^{-rt/2},
 \qquad\mbox{if}\;\underline{u}\,_t>Ke^x\;\;\mbox{and}\;\;
 x\in \Omega,
 \vspace{2mm} \\
 -{\cal L}\,\underline{u}\,_t\geq c-re^{-rt/2},
 \qquad\mbox{if}\;\underline{u}\,_t=Ke^x\;\;\mbox{and}\;\;
 x\in \Omega,
 \vspace{2mm} \\
 \underline{u}\,_t(0)\;=K.
 \end{array}
 \right.
 \ee
 We will give the explicit solution of the VI
 (\ref{eq5.6}) in Step 2. Denote
 $$
 w=\underline{u}\,_t-e^{rt/2-r\tau}+e^{-rt/2}.
 $$
Repeating the same argument as above, we deduce that
 \be                                                                             \label{eq5.7}
 u\geq w=\underline{u}\,_t-e^{rt/2-r\tau}+e^{-rt/2},
 \ee
provided $t$ is large enough.

 \noindent \underline{Step 2}: We solve the VIs (\ref{eq5.4}) and
 (\ref{eq5.6}). it is sufficient to solve the following elliptic VI:
 \be                                                                             \label{eq5.8}
 \left\{
 \begin{array}{l}
 -{\cal L}\,v=c^*,
 \qquad\mbox{if}\;v>Ke^x\;\;\mbox{and}\;\;
 x\in \Omega,
 \vspace{2mm} \\
 -{\cal L}\,v\geq c^*,
 \qquad\mbox{if}\;v=Ke^x\;\;\mbox{and}\;\;
 x\in \Omega,
 \vspace{2mm} \\
 v(0)\;=K.
 \end{array}
 \right.
 \ee
 It is clear that (\ref{eq5.4}) and (\ref{eq5.6}) coincide with
 the VI (\ref{eq5.8}) if we let $c^*=c+re^{-rt/2}$ and $c^*=c-re^{-rt/2}$,
 respectively.

 \noindent (1) In the case $c^*\leq rK(\alpha_+-1)/\alpha_+$, we
 first find out the bounded solution of the following associated free boundary problem
 of (\ref{eq5.8}):
 \be                                                                             \label{eq5.9}
 \left\{
 \begin{array}{l}
 -{\cal L}\,v=c^*>0,\quad x\in(-\infty,x^*),
 \vspace{2mm} \\
 \p_x v(x^*)=v(x^*)=Ke^{x^*}.
 \end{array}
 \right.
 \ee
It is not difficult to check that the solution of (\ref{eq5.9})
should
 take the form of
 $$
 v=A\,e^{\alpha_+x}+B\,e^{\alpha_-x}+{c^*\over r},\qquad
 x<x^*,
 $$
 where $\alpha_-$ is defined in Lemma \ref{lm4.1}. Since $v$ is bounded
 and $\alpha_-<0$, then we have $B=0$. Recalling the boundary condition,
 we deduce
 $$
 A\,e^{\alpha_+x^*}=Ke^{x^*}-{c^*\over r},\qquad
 A\,\alpha_+\,e^{\alpha_+x^*}=Ke^{x^*}.
 $$
Since $\alpha_+>1$, then we have
 \be                                                                             \label{eq5.10}
 x^*=\ln \left(\,{\alpha_+\over \alpha_+-1}\,{c^*\over rK}\,\right),
 \qquad
 v={K\over \alpha_+}\,e^{\alpha_+x+(1-\alpha_+)\,x^*}+{c^*\over r}.
 \ee
It is clear that $x^*\leq0$. Extend $v$ into $(-\infty,0\,]$ as
follows:
 \be                                                                             \label{eq5.11}
  v(x)=
 \left\{
 \begin{array}{ll}
 {\displaystyle{K\over \alpha_+}\,\exp\left\{\,\alpha_+x+(1-\alpha_+)\,x^*\,\right\}
 +{c^*\over r},}
 \quad &x<x^*,
 \vspace{2mm}\\
 Ke^x,  &x^*\leq x\leq 0.
 \end{array}
 \right.
 \ee

 Next, we prove that $v$ is the unique
 $W^2_{p,\,loc}(\Omega)\cap C(\,\overline{\Omega}\,)\cap L^\infty(\Omega)$
 solution of the VI (\ref{eq5.8}). In fact, the uniqueness follows from the comparison
 principle for VI {(see \cite{Yan1})}, and it is easy to verify the regularity of the solution.
 Then it is sufficient to prove that $v$ satisfies the VI (\ref{eq5.8}).
According to (\ref{eq5.10}), we can check that
 $$
 \p_x v(x)=K\,e^{\alpha_+x+(1-\alpha_+)\,x^*}
 =Ke^x\,e^{(\alpha_+-1)(x-x^*)}\leq Ke^x
 =\p_x(Ke^x),\quad x\leq x^*.
 $$
 By combining the boundary condition of (\ref{eq5.9}), we obtain that
 $$
 v(x)-Ke^x\geq 0,\;\forall\;x\leq x^*.
 $$
 Hence, we only need to prove that
 $$
 c^*\leq -{\cal L}\,Ke^x=q\,Ke^x,\;\;\forall\;\;x>x^*.
 $$
 It is sufficient to show that
 \be                                                                             \label{eq5.12}
 c^*\leq q\,Ke^{x^*}= q\,K\,{\alpha_+\over \alpha_+-1}\,{c^*\over rK}
 \Leftrightarrow  {\alpha_+\over \alpha_+-1}\geq {r\over q}
 \Leftrightarrow  \alpha_+ \leq {r\over r-q}.
 \ee
In fact, it is easy to check that
 $$
 {\sigma^2\over2}\,\left(\,{r\over r-q}\,\right)^2+
 \left(\,r-q-{\sigma^2\over2}\,\right)\,\left(\,{r\over r-q}\,\right)-r
 ={\sigma^2\over2}\,\left[\,\left(\,{r\over r-q}\,\right)^2
 -{r\over r-q}\,\right]>0.
 $$
 Recalling the definition of $\alpha_+$, we deduce (\ref{eq5.12}) from the
 property of quadratic functions. Hence, we have checked that $v$
 is the uniqueness solution of the VI (\ref{eq5.8}).

 \noindent (2) In the case of $c^*>rK(\alpha_+-1)/\alpha_+$, since
 $x^*$ defined in (\ref{eq5.10}) is larger than zero, then $v$ defined in
 (\ref{eq5.11}) is not the solution of the VI (\ref{eq5.8}). Now, we
 need to reconstruct the solution of the VI (\ref{eq5.8}). We first  solve
 the following ODE
 \be                                                                             \label{eq5.13}
 -{\cal L}\,v=c^*>0,\;\; x\in(-\infty,0);\qquad
 v(0)=K.
 \ee
 It is not difficult to check that the bounded solution is
 \be                                                                             \label{eq5.14}
 v(x)=Ke^{\alpha_+x}+{c^*\over r}\,\left(\,1-e^{\alpha_+x}\,\right),
 \quad x\leq 0.
 \ee

 Next, we prove that $v$ is the unique
 $W^2_{p,\,loc}(\Omega)\cap C(\,\overline{\Omega}\,)\cap L^\infty(\Omega)$
 solution of the VI (\ref{eq5.8}). By the same argument as above, it is sufficient to prove $v(x)\geq Ke^x$ for any $x\leq0$. Indeed,
we calculate
 $$
 \p_x v(x)=\alpha_+\,\left(\,K-{c^*\over r}\,\right)\,e^{\alpha_+x}
 \leq Ke^{\alpha_+x}\leq Ke^x,\;\;\forall\;x\leq0,
 $$
 where we have used $\alpha_+>1$. By combining the boundary condition of
 (\ref{eq5.13}), we deduce that $v(x)\geq Ke^x$ for any $x\leq0$.
 Hence, we have showed that $v$ is the unique solution of the VI (\ref{eq5.8}).

 \noindent \underline{Step 3}: Prove the property (1) in the case of
 $c<rK(\alpha_+-1)/\alpha_+$.

 In view of (\ref{eq5.5}) and (\ref{eq5.7}), we deduce the following
 inequality if $t$ is  large enough,
 $$
 \underline{u}\,_t(x)-e^{rt/2-r\tau}+e^{-rt/2}\leq u(x,\tau)
 \leq \overline{u}_t(x)+e^{rt/2-r\tau}-e^{-rt/2}.
 $$
 In particular, by taking $\tau=t$ we have
 $$
 \underline{u}\,_t(x)\leq u(x,t) \leq \overline{u}_t(x).
 $$
Since $\underline{u}\,_t,\,u,\,\overline{u}_t\geq Ke^x$, we derive
 $$
 \{x:\underline{u}\,_t(x)=Ke^x\}\supset
 \{x:u(x,t)=Ke^x\}\supset
 \{x:\overline{u}_t(x)=Ke^x\}.
 $$
It is not difficult to check that
 $$
 c+re^{-rt/2},\,c-re^{-rt/2}\rightarrow c<rK(\alpha_+-1)/\alpha_+\;\;
 \mbox{as}\;t\rightarrow+\infty.
 $$
 Hence, the conclusion in Step 2 implies that $\underline{u}\,_t,\,\overline{u}_t$
 takes the form of (\ref{eq5.11}) with $c^*=c-re^{-rt/2}$ and $c^*=c+re^{-rt/2}$,
 respectively.  Denote $\underline{x}\,_t,\,\overline{x}_t$ as the corresponding
 free boundary points $x^*$ defined in (\ref{eq5.10}). Since $t$ is arbitrary, then we have
 $$
 [\,\underline{x}\,_\tau,0\,]=\{x:\underline{u}\,_\tau(x)=Ke^x\}\supset
 \{x:u(x,\tau)=Ke^x\}\supset
 \{x:\overline{u}_\tau(x)=Ke^x\}=[\,\overline{x}_\tau,0\,],
 $$
 provided $\tau$ is large enough. Hence, the definition of the free boundary $c(\tau)$ implies that
 $$
 \ln \left(\,{\alpha_+\over \alpha_+-1}\,{c-re^{-r\tau/2}\over rK}\,\right)
 =\underline{x}\,_\tau\leq c(\tau) \leq \overline{x}_\tau
 =\ln \left(\,{\alpha_+\over \alpha_+-1}\,{c+re^{-r\tau/2}\over rK}\,\right)<0,
 $$
 provided $\tau$ is large enough.
Moreover, it is not difficult to check that
 $$
 \lim\limits_{\tau\rightarrow+\infty}\underline{x}\,_\tau=c_{\infty}
 =\lim\limits_{\tau\rightarrow+\infty}\overline{x}_\tau,\quad
 \lim\limits_{\tau\rightarrow+\infty}
 \underline{u}\,_\tau(x)= u_{1,\,\infty}(x)
 =\lim\limits_{\tau\rightarrow+\infty}
 \overline{u}\,_\tau(x),\;
 \forall\;x\in\overline{\Omega}.
 $$
 Hence, the property (1) follows.

 \noindent\underline{Step 4}: Prove the property (1) in the case of
 $c=rK(\alpha_+-1)/\alpha_+$.

 In this case, $c-re^{-rt/2}<rK(\alpha_+-1)/\alpha_+$, and
 $\underline{u}\,_t$ still takes the form of (\ref{eq5.11}) if $t$ is
 large enough. Repeating same the argument as in Step 3, we still
 have that
 \bee
 &&c(\tau)\geq\underline{x}\,_\tau=
 \ln \left(\,{\alpha_+\over \alpha_+-1}\,{c-re^{-r\tau/2}\over rK}\,\right),
 \qquad\liminf\limits_{\tau\rightarrow+\infty} c(\tau)\geq
 \lim\limits_{\tau\rightarrow+\infty} \underline{x}\,_\tau
 =0=c_{\infty},
 \\[2mm]
 &&\liminf\limits_{\tau\rightarrow+\infty} u(x,\tau)
 \geq \lim\limits_{\tau\rightarrow+\infty} \underline{u}\,_\tau(x)
 =u_{1,\,\infty}(x)
 ={K\over \alpha_+}\,e^{\,\alpha_+x}+{c\over r},
 \quad\forall\;x\in\overline{\Omega},
 \eee
 provided $\tau$ is large enough.

 On the other hand, the definition of the free boundary $c(\tau)$
 implies that $c(\tau)\leq0$. Hence, we deduce that
 $$
 \lim\limits_{\tau\rightarrow+\infty} c(\tau)=0=c_{\infty}.
 $$
By applying the same method as in Step 3, we derive that
 $$
 \limsup\limits_{\tau\rightarrow+\infty} u(x,\tau)
 \leq \lim\limits_{\tau\rightarrow+\infty} \overline{u}_\tau(x),
 \quad\forall\;x\in\overline{\Omega}.
 $$
Since
$c+re^{-r\tau/2}>(\alpha_+-1)\,rK/\alpha_+,\,\overline{u}_\tau$
takes form of (\ref{eq5.14}) rather than (\ref{eq5.11}).
 It is easy to calculate that
 $$
 \lim\limits_{\tau\rightarrow+\infty}\overline{u}_\tau
 =\left(\,K-{c\over r}\,\right)\,e^{\alpha_+x}+{c\over r}
 =\left(\,K-{(\alpha_+-1)\,rK\over \alpha_+r}\,\right)
 \,e^{\alpha_+x}+{c\over r}=u_{1,\,\infty}(x).
 $$
From the above arguments, we have that
 $$
 \liminf\limits_{\tau\rightarrow+\infty} u(x,\tau)
 \geq \lim\limits_{\tau\rightarrow+\infty}\underline{u}\,_\tau(x)
 =u_{1,\,\infty}(x)
 =\lim\limits_{\tau\rightarrow+\infty}\overline{u}_\tau(x)
 \geq\limsup\limits_{\tau\rightarrow+\infty} u(x,\tau).
 $$
 Hence, we have proved the property (1) in the case of $c=rK(\alpha_+-1)/\alpha_+$.

 \noindent \underline{Step 5}: Prove the property (2).

 In this case, $\underline{u}\,_\tau,\,\overline{u}_\tau$ take the form
 of (\ref{eq5.14}) if $\tau$ is large enough.
 Repeating the same arguments as in Step 3,
 we get
 $$
 \{x=0\}=\{x:\underline{u}\,_\tau(x)=Ke^x\}\supset
 \{x:u(x,\tau)=Ke^x\}\supset
 \{x:\overline{u}_\tau(x)=Ke^x\}=\{x=0\},
 $$
 provided $\tau$ is large enough. Then the definition of the free boundary
 $c(\tau)$ implies that $c(\tau)=0$ if $\tau$ is large enough. Hence, there
 exists a positive constant $\overline{T}$ such that
 $$
 c(\tau)=0,\;\;\forall\;\tau\geq \overline{T}.
 $$
It is clear that
 $$
 \lim\limits_{\tau\rightarrow+\infty}\underline{u}\,_\tau(x)
 =u_{2,\,\infty}(x)
 =\lim\limits_{\tau\rightarrow+\infty}\overline{u}_\tau(x),\;\;
 \forall\;x\in\overline{\Omega}.
 $$
 Hence, the property (2) follows.
\end{proof}

 In view of the properties (2), (3) in Theorem \ref{th5.2} and the property (1) in
 Theorem \ref{th5.3}, we claim the non-monotonicity property of the free boundary $c(\tau)$ (see Figure 3.8).

\begin{thm}  (Non-monotonicity of the free boundary)

                                                                       \label{th5.4}
 For the case $c<qK$, if furthermore $c\leq
 rL(\alpha_+-1)/\alpha_+$ holds, then the free boundary $c(\tau)$ is non-monotonic in the interval
$[\,0,T\,]$ (where we suppose that $T$ is large enough).
\end{thm}
\begin{proof}
 If $c_0\geq c_{\infty}$, then the properties (2),
 (3) in Theorem \ref{th5.2} imply that there exists a $t_1>0$ such that
 $$
 c(t_1)>c_0=c(0)\geq c_{\infty}.
 $$
According to the property (1) in Theorem \ref{th5.3}, we know that
there exists a
 $t_2$ large enough such that $t_2>t_1$ and
 $$
 c(t_2)\leq {c_{\infty}+c(t_1)\over 2}<c(t_1).
 $$
 Hence, the free boundary $c(\tau)$ is non-monotonic. On the other hand, it is
 clear that
 \bee
 c_{0}\geq c_{\infty}\Leftrightarrow
 \max\left\{\,{c\over qK},\,{L\over K}\right\}\geq
 {\alpha_+\over \alpha_+-1}\;{c\over rK}
 \Leftrightarrow
 \max\left\{\,{c\over q},\,L\right\}\geq
 {\alpha_+\over \alpha_+-1}\;{c\over r}.
 \eee
 By applying the same method as in the proof of (\ref{eq5.12}), we conclude that
 $$
 \alpha_+ <{r\over r-q}\Leftrightarrow (r-q)\alpha_+< r
 \Leftrightarrow {c\over q}<{\alpha_+\over \alpha_+-1}\;{c\over r}.
 $$
 Hence,
 $$
 c_{0}\geq c_{\infty}\Leftrightarrow
 L\geq {\alpha_+\over \alpha_+-1}\;{c\over r}
 \Leftrightarrow
 c\leq {rL\,(\alpha_+-1)\over \alpha_+}.
 $$
\end{proof}

 Next, we consider the monotonicity and regularity of the free boundary $c(\tau)$ if $c\geq rL$. Since (\ref{eq4.13}) holds,
 the problem is
relatively standard in this case.

\begin{thm} (Regularity of the free boundary)                                                                     \label{th5.5}
 For the case $c<qK$, if furthermore $c\geq rL$ holds, the free boundary $c(\tau)$ is increasing with respect
 to $\tau$ on the interval $[\,0,T\,]$ with $c(\tau)\in C[\,0,T\,]\cap C^\infty(0,T\,]$.
 {Moreover, $c(\tau)$ is strictly increasing on $[\,0,\underline{T}\,]$ with
 $\underline{T}=\sup\{\tau\in[\,0,T\,]:c(\tau)<0\}$.}
\end{thm}
\begin{proof}
 According to (\ref{eq4.12}) and (\ref{eq4.13}), we have
 \bee
 \p_x(u-Ke^x)\leq0,\;\;\p_\tau(u-Ke^x)\geq0\;\;\mbox{a.e. in }\Omega_T.
 \eee
 By combining $u-Ke^x\in C(\,\overline{\Omega_T}\,)$, we deduce
 that $u(x,\tau)-Ke^x$ is increasing with respect to $\tau$ and decreasing
 with respect to $x$.

 For any fixed $\tau_0\in(0,T\,]$ and any $x\in[\,c(\tau_0),0\,],\;\tau\in[\,0,\tau_0\,]$, we derive that
 $$
 0\leq u(x,\tau)-Ke^x\leq u(c(\tau_0),\tau)-Ke^{c(\tau_0)}
 \leq  u(c(\tau_0),\tau_0)-Ke^{c(\tau_0)}=0,
 $$
 where we have used
 that $u=Ke^x$ on the free boundary. Hence, the definition of the free boundary
 implies that $c(\tau)\leq c(\tau_0)$ for any $\tau\in[\,0,\tau_0\,]$. Hence,
 we deduce that $c(\tau)$ is increasing on $[\,0,T\,]$.

 The property (3) in Theorem \ref{th5.2} implies that $c(\tau)$ is right-continuous
 at $\tau=0$. Next, we prove that $c(\tau)$ is continuous on $(0,T\,]$. Otherwise,
 there exist some constants $x_1,\,x_2,\,t_1$ such that $x_2<x_1\leq 0,\,0<t_1<T,\,
 \lim_{\tau\ra t_1^-}c(\tau)=x_2,\;\lim_{\tau\ra t_1^+}c(\tau)=x_1$ (see Figure 3.9), and
 $$
 \p_\tau u-{\cal L}u=c\;\;\mbox{in}\;\;
 (x_2,x_1)\times[\,t_1,T\,],\qquad
 u(x,t_1)=Ke^x,\;\forall\;x\in(x_2,x_1).
 $$
  \begin{picture}(270,130)(90,0)
 \put(240,10){\vector(1,0){140}} \put(360,6){\vector(0,1){100}}
 \put(240,95){\line(1,0){125}}
 \put(358,108){$\tau$}\put(383,10){$x$}
 {\thicklines\qbezier(260,9)(290,15)(300,30)
 \put(300,40){\line(1,0){30}}
 \qbezier(330,40)(350,42)(355,95)}
 \multiput(300,95)(0,-5){17}{\line(0,-1){2}}
 \multiput(330,95)(0,-5){17}{\line(0,-1){2}}
 \put(328,7){$\bullet$}\put(328,-3){$x_1$}
 {\thicklines\put(300,30){\line(0,1){10}}}
 \put(297,7){$\bullet$}\put(298,-3){$x_2$}
 \multiput(330,40)(5,0){7}{\line(1,0){2}}
 \multiput(300,30)(5,0){12}{\line(1,0){2}}
 \put(358,37){$\bullet$}\put(368,37){$t_1$}
 \put(358,27){$\bullet$}\put(368,27){$t_2$}
 \put(358,92){$\bullet$}\put(368,90){$T$}
 \put(265,45){${\bf CT}_x$}\put(321,20){${\bf CV}_x$}
 \put(305,45){$c(\tau)$}\put(307,75){${\cal M}$}
 \end{picture}
 \begin{center}
 Figure$\!$ 3.9. $\;\;\;$Non-continuous free boundary
 \end{center}

If $x_2<x<x_1$, then we have
 $$
 \p_\tau u(x,t_1)=c+{\cal L} Ke^x=c-qK\,e^x<0,
 $$
 where the last inequality follows from $x_2\geq \underline{X}\,$, which is deduced from
 the property (1) in Theorem \ref{th5.2}. It is clear that the above inequality
 contradicts (\ref{eq4.13}).


 {Next, we prove that $c(\tau)$ is strictly increasing on
 $[\,0,\underline{T}\,]$.} Otherwise, there exist some constants
 $x_2,\,t_1,\,t_2$ such that $x_2<0,\;0\leq t_2<t_1\leq {\underline{T}}$ and
 $c(\tau)=x_2$ for any $\tau\in[\,t_2,t_1\,]$ (see Figure 3.9).
It is clear that $u(x,\tau)=Ke^x$ for any
$(x,\tau)\in[\,x_2,0\,]\times[\,t_2,t_1\,]$.
 Since $\p_x u$ continuously crosses the free boundary, then
 $\p_x u(x_2,\tau)=Ke^{x_2}$ for any $\tau\in[\,t_2,t_1\,]$.
 We then deduce that
 \be                                                                          \label{eq5.15}
 \p_\tau u(x_2,\tau)=0,\;\;\p_\tau(\p_xu)(x_2,\tau)=0,\quad
 \forall\;\tau\in[\,t_2,t_1\,].
 \ee
 On the other hand, in the domain ${\cal N}=(-\infty,x_2)\times(t_2,t_1\,]$,
 $u$ and $\p_\tau u$ respectively satisfies
 \bee
 &&\p_\tau u-{\cal L}u=c\;\;\;\mbox{in}\;\;{\cal N},\qquad
 u(x_2,\tau)=Ke^{x_2},\;\forall\;\tau\in (t_2,t_1),
 \\[2mm]
 &&\left\{
 \begin{array}{l}
 \p_\tau (\p_\tau u)-{\cal L}(\p_\tau u)=0,\;\;\;
 \p_\tau u\geq0\;\;\;\mbox{in}\;\;{\cal N},
 \vspace{2mm}\\
 \p_\tau u(x_2,\tau)=0,\;\forall\;\tau\in (t_2,t_1).
 \end{array}
 \right.
 \eee
 By applying the Hopf lemma, we deduce $\p_x(\p_\tau u)(x_2,\tau)<0$,
 which contradicts the second equality in (\ref{eq5.15}).

Finally, since we have the estimate (\ref{eq4.13}), it is standard
to show that $C^\infty(0,T\,]$ (see \cite{Fr2}).
\end{proof}

 Next, we improve the regularity of the free boundary $c(\tau)$ for the case
 $c<rL$. In this case, (\ref{eq4.13}) is false, so the standard method in
 \cite{Fr2} does not apply to this problem. The main idea to improve the regularity
 is to apply some proper coordinate transformation to the original problem, and
 transform it into a new problem, and achieve the estimate similar to (\ref{eq4.13}).

\begin{thm} (Regularity of the free boundary)

                                                                         \label{th5.6}
 For the case $c<qK$, if furthermore $c<rL$ holds, then the free boundary $c(\tau)\in C[\,0,T\,]\cap C^\infty(0,T\,]$.
\end{thm}
\begin{proof} We first apply the following transformation
 \be                                                                                           \label{eq5.16}
 y=x+\left(\,r-\frac cL\,\right)\tau,\qquad\;
 v(y,\tau)=e^{(r-\frac cL)\,\tau}\,(\,u(x,\tau)-Ke^x\,).
 \ee
 It is not difficult to deduce that $v$ satisfies the following VI (see Figure 3.10):

\begin{picture}(300,130)(100,0)
 \put(190,10){\vector(1,0){190}} \put(260,6){\vector(0,1){105}}
 \put(258,113){$\tau$}\put(383,10){$y$}
 {\thicklines\qbezier(210,9)(265,50)(280,100)}
 \put(258,7){$\bullet$}\put(252,-4){$O$}
 \put(258,13){$\bullet$}\put(252,16){$\delta$}
 \put(260,10){\line(1,1){90}}
 \qbezier[100](190,15)(225,15)(260,15)
 \qbezier[100](260,15)(300,55)(345,100)
 \put(220,45){${\bf CT}_y$}\put(265,45){${\bf CV}_y$}
 \put(245,78){${\bf c_y(\tau)}$}
 \end{picture}
 \begin{center}
 Figure$\!$ 3.10.$\;\;\;\;\;$The free boundary $c_y(\tau)$ after transformation
\end{center}
 \be                                                                                           \label{eq5.17}
 \left\{
 \begin{array}{ll}
 \p_\tau v-{\cal L}_y v
 =ce^{(r-\frac cL)\tau}-q\,\!Ke^{y},
 \qquad\mbox{if}\;v>0\;\;\mbox{and}\;\;
 (y,\tau)\in \Omega\,^y_T,
 \vspace{2mm}\\
 \p_\tau v-{\cal L}_y v
 \geq ce^{(r-\frac cL)\tau}-q\,\!Ke^{y},
 \qquad\mbox{if}\;v=0\;\;\mbox{and}\;\;
 (y,\tau)\in \Omega\,^y_T,
 \vspace{2mm}\\
 v((r-c/L)\tau,\tau)=0,\hspace{3.8cm}0\leq\tau\leq T,
 \vspace{2mm}\\
 v(y,0)=(L-Ke^y)^+,\hspace{3.8cm}y\leq0,
 \end{array}
 \right.
 \ee
 where
 $$
 {\cal L}_y v=\frac {\sigma^2}{2}\,\p_{yy}v +\left(\,\frac cL -q-\frac
 {\sigma^2}{2}\,\right)\,\p_y v-\frac cL v,\;\;
 \Omega\,^y_T\stackrel{\Delta}{=}
 \left\{\,y<\left(\,r-\frac cL\,\right)\,\tau,\;
 0<\tau\leq T\,\right\}.
 $$

 For any small enough $\delta>0$, we denote
 $$
 \widetilde{v}(y,\tau)=v(y,\tau+\delta),\;\;(y,\tau)\in \Omega\,^y_{T-\delta}.
 $$
 Then $\widetilde{v}$ satisfies the following VI (see Figure 3.10):
 \bee
 \left\{
 \begin{array}{ll}
 \p_\tau \widetilde{v}-{\cal L}_y \widetilde{v}
 =ce^{(r-\frac cL)(\tau+\delta)}-q\,\!Ke^{y}\geq ce^{(r-\frac
 cL)\tau}-q\,\!Ke^{y},
 \quad&\mbox{if}\;\widetilde{v}>0\;\;\mbox{and}\;\;
 (y,\tau)\in \Omega\,^y_{T-\delta},
 \vspace{2mm}\\
 \p_\tau \widetilde{v}-{\cal L}_y \widetilde{v}
 \geq ce^{(r-\frac cL)(\tau+\delta)}-q\,\!Ke^{y}\geq ce^{(r-\frac
 cL)\tau}-q\,\!Ke^{y},
 \quad&\mbox{if}\;\widetilde{v}=0\;\;\mbox{and}\;\;
 (y,\tau)\in \Omega\,^y_{T-\delta},
 \vspace{2mm}\\
 \widetilde{v}((r-c/L)\tau,\tau)=v((r-c/L)\tau,\tau+\delta)\geq0,&0\leq\tau\leq
 T-\delta,
 \vspace{2mm}\\
 \widetilde{v}(y,0)=v(y,\delta),&y\leq0.
 \end{array}
 \right.
 \eee

 Next, we prove $\widetilde{v}\geq v$ in $\overline{\Omega\,^y_{T-\delta}}$. In fact, the comparison
 principle for VI {(see \cite{Yan1})} implies that it is sufficient to show that
 $$
 \widetilde{v}(y,0)=v(y,\delta)\geq (L-Ke^y)^+=v(y,0).
 $$
Moreover, since $v\geq0$, then what we need to prove is that
$L-Ke^y$ is a subsolution of~\eqref{eq5.17}.
 Indeed, we can check that
 \bee
 \left\{
 \begin{array}{ll}
 \p_\tau (L-Ke^y)-{\cal L}_y (L-Ke^y)=c-q\,\!Ke^{y}\leq
 ce^{(r-\frac cL)\tau}-q\,\!Ke^{y},
 \vspace{2mm}\\
 (L-Ke^y)\Big|_{y=(r-c/L)\tau}\leq0=v(y,\tau)\Big|_{y=(r-c/L)\tau},\;\;0\leq\tau\leq T.
 \end{array}
 \right.
 \eee
Hence, we conclude that $L-Ke^y$ is indeed a subsolution
of~\eqref{eq5.17}.

We have showed $v(y,\tau+\delta)=\widetilde{v}(y,\tau)\geq
v(y,\tau)$ in $\overline{\Omega\,^y_{T-\delta}}$
 for any small enough $\delta$, which implies $\p_\tau v\geq0$ almost everywhere in $\Omega\,^y_T$. Hence, by using the
 method as in \cite{Fr2}, we can prove that $c_y(\tau)\in C[\,0,T\,]\cap C^\infty(0,T\,]$.
According to the transformation (\ref{eq5.16}), we have
 $c_x(\tau)=c_y(\tau)-\left(\,r-\frac cL\,\right)\tau$. Therefore,
 $c(\tau)\in C[\,0,T\,]\cap C^\infty(0,T\,]$.
\end{proof}

\appendix
\section{The Proof of Theorem \ref{th4.1}}

We prove Theorem \ref{th4.1} in this appendix. Since (\ref{eq4.1})
lies in the unbounded domain $\Omega_T$,
 we use the following VI in the bounded domain to approximate
 (\ref{eq4.1}),
 \be                                                                               \label{eq4.2}
 \left\{
 \begin{array}{l}
 \p_\tau u_n-{\cal L} u_n=c,
 \qquad\qquad\qquad\mbox{if}\;u_n>Ke^x\;\;\mbox{and}\;\;
 (x,\tau)\in \Omega_T^n,
 \vspace{2mm} \\
 \p_\tau u_n-{\cal L} u_n\geq c,
 \qquad\qquad\qquad\mbox{if}\;u_n=Ke^x\;\;\mbox{and}\;\;
 (x,\tau)\in \Omega_T^n,
 \vspace{2mm} \\
 u_n(-n,\tau)\;={c\over r}+{rL-c\over r}\,e^{-r\tau},
 \qquad u_n(0,\,\tau)\;=K,\qquad
 0\leq \tau\leq T,
 \vspace{2mm} \\
 u_n(x,0)\;=\max\{L,Ke^x\},\hspace{1cm}\qquad\quad
 -n\leq x\leq 0,
 \end{array}
 \right.
 \ee
 where $\Omega_T^n\stackrel{\Delta}{=}(-n,0)\times(0,T\,]$ and
 $n\in \mathbb{N}_+$ satisfying $n>\max\{\,\ln K-\ln L,\ln r+\ln K-\ln c\,\}$.

 Next, we utilize the penalty method to prove the existence of the solution of
 (\ref{eq4.2}). We first  construct the penalty function
 $\beta_\ep(\cdot)$ such that
\be\nonumber
 &\beta_\ep (s)\in C^\infty(\mathbb{R}),\qquad \beta_\ep(s)\geq 0,\qquad
 \beta_\ep^{\prime }(s)\geq 0,\qquad \beta_\ep^{\prime \prime }(s)\geq
 0,
 \\                                                                              \label{eq4.3}
 &\beta_\ep(s)=0\;\,\mbox{for any}\;\,s\leq-\ep,\qquad\qquad
 \beta_\ep(0)=M\stackrel{\triangle}{=}qK-c>0,
 \ee
 and
 $$
    \lim\limits_{\varepsilon\rightarrow 0} \beta_\varepsilon (s)
    =\left\{
    \begin{array}{ll}
    0, & s<0,
    \vspace{2mm}\\
    +\infty, &s>0.
    \end{array}
    \right.
 $$

 Then we use the following penalty problem to approximate (\ref{eq4.2}):
\be \label{eq4.4}
 \left\{
 \begin{array}{l}
 \p_\tau u_{\ep,\,n}-{\cal L} u_{\ep,\,n}-\beta_\ep(Ke^x-u_{\ep,\,n})=c
 \qquad\quad\mbox{in}\;\;\Omega_T^n,
 \vspace{2mm} \\
 u_{\ep,\,n}(-n,\,\tau)\;={c\over r}+{rL-c\over r}\,e^{-r\tau},\qquad
 u_{\ep,\,n}(0,\,\tau)\;=K,\hspace{0.6cm}
 0\leq \tau\leq T,
 \vspace{2mm} \\
 u_{\ep,\,n}(x,0)\;=\pi_\ep(Ke^x-L)+L,\hspace{1cm}\quad\;\;
 -n\leq x\leq 0,
 \end{array}
 \right.
\ee
 where $\pi_\ep(s)$ is a smoothing function for smoothing the initial value
 $\max\{L,Ke^x\}$,  which satisfies
 $\pi_\ep (s)\in C^\infty(\mathbb{R}),\;\;\pi_\ep(s)\geq s,\;\;
 0\leq \pi_\ep^\prime(s)\leq 1,\;\;
 \pi_\ep ^{\prime\prime}(s)\geq 0,\;\;
 \lim\limits_{\ep\rightarrow 0^+}\pi_\ep (s)=s^+$ and
 \bee
 \pi_\ep (s)=\left\{
 \begin{array}{ll}
 s, & s\geq\ep,   \vspace{2mm}\\
 0, & s\leq -\ep.
 \end{array}
 \right.
 \eee

\begin{lem}                                                                             \label{lm4.1}
 For any fixed $n$ and $\ep$, (\ref{eq4.4}) has a unique strong solution
 such that $u_{\ep,n}\in W_p^{2,\,1}(\Omega_T^n)\cap C(\,\overline{\Omega_T^n}\,)$ for
 any  $1<p<\infty$, and we have the following estimates:
 \be                                                                                     \label{eq4.5}
 &\displaystyle{\max\left\{Ke^x,{c\over r}
 +{rL-c\over r}\,e^{-r\tau}\right\}
 \leq u_{\ep,\,n}\leq K}
 &\mbox{in}\;\;\overline{\Omega_T^n};
 \\[2mm]                                                                                 \label{eq4.6}
 &0\leq\p_x  u_{\ep\,n}
 \leq K\left(e^{x}-\alpha_- e^{\alpha_-(x+n)-n}\right)
 &\mbox{on}\;\;\overline{\Omega_T^n},
 \ee
 where $\alpha_+,\,\alpha_-$ are the
 positive and negative characteristic roots for the ordinary differential
 operator ${\cal L}$, respectively. That is, $\alpha_+,\,\alpha_-$
 are respectively
 the positive and negative roots of the following algebra equation:
 $$
 {\sigma^2\over2}\,\alpha^2+
 \left(\,r-q-{\frac{\sigma ^2}2}
 \,\right)\,\alpha-r=0.
 $$
 If furthermore $c\geq r\,L$, we have the following estimate:
 \be                                                                                   \label{eq4.7}
 \p_\tau  u_{\ep\,n} \geq -r\ep\;\;
 \mbox{a.e. in }\;\;\Omega_T^n.
 \ee
\end{lem}
\begin{proof}
 The existence of the solution to (\ref{eq4.4}) can be
proved in a similar way as in \cite{Yan,Yang3,Yang2}, and we refer
to those papers for the details. The uniqueness follows directly
from the A-B-P maximum principle (see \cite{Tao}).


 Next, we prove (\ref{eq4.5}). Letting $w=Ke^x$ and
 recalling (\ref{eq4.3}), we calculate that
 \bee
 \left\{
 \begin{array}{l}
 \p_\tau w-{\cal L} w-\beta_\ep(Ke^x-w)-c
 =qKe^x-\beta_\ep(0)-c
 =qKe^x+c-qK-c\leq0,
 \vspace{2mm} \\
 w(-n,\,\tau)=Ke^{-n}\leq \min\left\{{c\over r},L\right\}
 \leq u_{\ep,\,n}(-n,\,\tau),
 \quad w(0,\,\tau)=K=u_{\ep,\,n}(0,\,\tau),
 \vspace{2mm} \\
 w(x,0)=Ke^x\leq \max\{L,Ke^x\}\leq \pi_\ep(Ke^x-L)+L=u_{\ep,\,n}(x,0).
 \end{array}
 \right.
 \eee
 Hence, $w=Ke^x$ is a sub-solution of (\ref{eq4.4}), and
 we have showed
 $u_{\ep,\,n}\geq Ke^x$. Letting
 $$w={c\over r}+{rL-c\over r}\,e^{-r\tau},$$
 we have that
 \bee
 \left\{
 \begin{array}{l}
 \p_\tau w-{\cal L} w-\beta_\ep(Ke^x-w)-c
 \leq \p_\tau w+rw-c=0,
 \vspace{2mm} \\
 w(-n,\,\tau)=u_{\ep,\,n}(-n,\,\tau),
 \;\; w(0,\,\tau)\leq\max\left\{{c\over r},L\right\}
 \leq\max\left\{{qK\over r},L\right\}
 \leq K=u_{\ep,\,n}(0,\,\tau),
 \vspace{2mm} \\
 w(x,0)=L\leq \max\{L,Ke^x\}\leq \pi_\ep(Ke^x-L)+L=u_{\ep,\,n}(x,0).
 \end{array}
 \right.
 \eee
 Therefore, $w$ is another sub-solution of (\ref{eq4.4}), and we
 have proved the first inequality in (\ref{eq4.5}).

%
 Moreover, it is easy to check that $K$ is a super-solution of
 (\ref{eq4.4}). Hence,  the second inequality in (\ref{eq4.5}) is obvious.

 Next, we prove the second inequality in (\ref{eq4.6}). Let
 $$W={c\over r}+{rL-c\over r}\,e^{-r\tau}+K\left(e^{x}-e^{\alpha_-(x+n)-n}\right).$$
 If $\ep$ is small enough and $n$ is large enough, then in the domain $\Omega_T^n,\,W$ satisfies
 \bee
 && W(x,\tau)\geq \min\left\{{c\over r},L\right\}+Ke^x-Ke^{-n}
 \geq Ke^x+\ep,
 \\[2mm]
 &\mbox{and}&\left\{
 \begin{array}{l}
 \p_\tau W-{\cal L} W+\beta_\ep(Ke^x-W)-c
 =c+qKe^x-c>0,
 \vspace{2mm} \\
 W(-n,\,\tau)= u_{\ep,\,n}(-n,\,\tau);
 \quad W(0,\,\tau)\geq K+\ep>K=u_{\ep,\,n}(0,\,\tau),
 \vspace{2mm} \\
 W(x,0)\geq L+Ke^x-Ke^{-n}\geq \pi_\ep(Ke^x-L)+L=u_{\ep,\,n}(x,0).
 \end{array}
 \right.
 \eee
 Hence, $W$ is another super-solution of (\ref{eq4.4}), and satisfies
 $$
 u_{\ep,\,n}(x,\tau)\leq
 {c\over r}+{rL-c\over r}\,e^{-r\tau}+K\left(e^{x}-e^{\alpha_-(x+n)-n}\right)
 =u_{\ep,\,n}(-n,\tau)+K\left(e^{x}-e^{\alpha_-(x+n)-n}\right).
 $$
If we define
 $$
 \overline{W}(x,\tau)=K\left(e^{x}-\alpha_- e^{\alpha_-(x+n)-n}\right),
 $$
 then we have $\p_x u_{\ep,\,n}(-n,\,\tau)\leq \overline{W}(-n,\,\tau)$.
Since $u_{\ep,\,n}(x,\,\tau)\geq Ke^x$ while $x\leq0$,
 and $u_{\ep,\,n}(0,\,\tau)=Ke^x|_{x=0}$, we conclude that
 $$\p_x u_{\ep,\,n}(0,\,\tau)\leq K e^x\Big|_{x=0}\leq \overline{W}(0,\,\tau).$$
Differentiating (\ref{eq4.4}) with respect to $x$, we deduce that
\bee
 \left\{
 \begin{array}{l}
 (\p_\tau -{\cal L})(\p_x u_{\ep,\,n}-\overline{W})
 +\beta'_\ep(\cdot)(\p_x u_{\ep,\,n}-\overline{W})
 =-(\p_\tau \overline{W}-{\cal L} \overline{W})
 +\beta'_\ep(\cdot)(Ke^x-\overline{W})
 \vspace{2mm} \\
 \qquad\qquad\qquad\qquad\qquad\quad\qquad\qquad\qquad\quad\;\;\;
 \leq-(\p_\tau \overline{W}-{\cal L} \overline{W})
 =-qKe^x<0,
 \vspace{2mm} \\
 \p_x u_{\ep,\,n}(-n,\,\tau)-\overline{W}(-n,\,\tau)\;\leq0,
 \qquad \p_x u_{\ep,\,n}(0,\,\tau)-\overline{W}(0,\,\tau)\;\leq 0,
 \vspace{2mm} \\
 \p_x u_{\ep,\,n}(x,0)-\overline{W}(x,\,0)\;
 =\pi'_\ep(Ke^x-L)Ke^x-\overline{W}(x,\,0)
 \leq Ke^x-\overline{W}(x,\,0)\leq 0.
 \end{array}
 \right.
 \eee
Hence, the comparison principle implies the second inequality in
(\ref{eq4.6}).

 Recalling (\ref{eq4.5}) and the boundary condition in (\ref{eq4.4}), we deduce
 that for any $\tau\in[\,0,T\,]$, the following inequalities hold
 $$
 \p_x u_{\ep,\,n}(0,\,\tau)\geq 0,\qquad
 \p_x u_{\ep,\,n}(-n,\,\tau)\geq 0.
 $$
Differentiating (\ref{eq4.4}) with respect to $x$, we derive that
 \bee
 \left\{
 \begin{array}{l}
 (\p_\tau -{\cal L})\p_x u_{\ep,\,n}
 +\beta'_\ep(Ke^x-u_{\ep,\,n})\,\p_x u_{\ep,\,n}
 =\beta'_\ep(Ke^x-u_{\ep,\,n})\,Ke^x\geq 0,
 \vspace{2mm} \\
 \p_x u_{\ep,\,n}(-n,\,\tau)\;\geq0,\qquad
 \p_x u_{\ep,\,n}(0,\,\tau)\;\geq0,
 \vspace{2mm} \\
 \p_x u_{\ep,\,n}(x,0)\;
 =\pi'_\ep(Ke^x-L)\,Ke^x\geq0.
 \end{array}
 \right.
 \eee
 Hence, the comparison principle implies the first inequality in (\ref{eq4.6}).

 In order to prove (\ref{eq4.7}), we differentiate (\ref{eq4.4}) with respect to
 $\tau$, then we have
 \bee
 \left\{
 \begin{array}{l}
 (\p_\tau -{\cal L})\p_\tau u_{\ep,\,n}
 +\beta'_\ep(Ke^x-u_{\ep,\,n})\,\p_\tau u_{\ep,\,n}=0,
 \vspace{2mm} \\
 \p_\tau u_{\ep,\,n}(-n,\,\tau)\;=
 (c-rL)\,e^{-r\tau}\geq0,\qquad
 \p_\tau u_{\ep,\,n}(0,\,\tau)\;=0.
 \end{array}
 \right.
 \eee
 Recalling (\ref{eq4.4}), we deduce that
 \bee
 &&\p_\tau u_{\ep,\,n}(x,0)\;
 =c+{\cal L}u_{\ep,\,n}(x,0)+\beta_\ep(Ke^x-u_{\ep,\,n}(x,0))
 \\[2mm]
 &\geq& c+(r-q)\,\pi'_\ep(Ke^x-L)\,Ke^x
 -rL-r\pi_\ep(Ke^x-L)+\beta_\ep(Ke^x-L-\pi_\ep(Ke^x-L))
 \\[2mm]&\geq&
 \left\{
 \begin{array}{l}
 c-rL\geq0,\hspace{8.5cm} Ke^x-L<-\ep,
 \vspace{2mm} \\
 c-rL-r\ep\geq-r\ep,\hspace{6.7cm} -\ep\leq Ke^x-L\leq \ep,
 \vspace{2mm} \\
 c+(r-q)\,Ke^x-rL-r(Ke^x-L)+qK-c=qK-qKe^x\geq0,\;\; Ke^x-L>\ep.
 \end{array}
 \right.
 \eee
 Moreover, it is clear that
 $$
 (\p_\tau -{\cal L})(-r\ep)
 +\beta'_\ep(Ke^x-u_{\ep,\,n})\,(-r\ep)\leq-r^2\ep<0.
 $$
 Hence, (\ref{eq4.7}) follows from the comparison principle.
 \end{proof}

\begin{lem}                                                                              \label{lm4.2}
 For any fixed $n\in I\!\!N$ satisfying $n>\max\{\ln K-\ln L,\,\ln r+\ln K-\ln c\}$,
 (\ref{eq4.2}) has a unique solution
 $u_n\in W_p^{2,1}(\Omega_T^n \backslash B_\delta(P_0))\cap C(\overline{\Omega_T^n})$
 for any $1<p<+\infty$, where $P_0=(-\ln K+\ln L,0),\;
 B_\delta(P_0)=\{(x,t):(x+\ln K-\ln L)^2+t^2\leq\delta^2\}$.
 Moreover, $\p_x u_n\in C(\Omega_T)$ and we have the following estimates:
 \be                                                                               \label{eq4.8}
 &\displaystyle{\max\left\{Ke^x,{c\over r}
 +{rL-c\over r}\,e^{-r\tau}\right\}
 \leq u_n\leq K}\;\;\;
 &\mbox{in}\;\;\overline{\Omega^n_T};
 \\[2mm]                                                                            \label{eq4.9}
 &0\leq\p_x  u_n
 \leq K\left(e^{x}-\alpha e^{\alpha_-(x+n)-n}\right)\;\;\;
 &\mbox{in}\;\;\overline{\Omega^n_T},
 \ee
 where $\alpha_-$ is defined in Lemma \ref{lm4.1}. If furthermore
 $c\geq r\,L$ holds, we have the following estimate:
 \be                                                                                \label{eq4.10}
 \p_\tau  u_n \geq 0\;\;\;
 \mbox{a.e. in}\;\;\Omega_T^n.
 \ee
\end{lem}
\begin{proof} From (\ref{eq4.3}) and (\ref{eq4.5}), we deduce that
 $$
 0\leq\beta_\ep(Ke^x-u_{\ep,\,n})\leq \beta_\ep(0)=M.
 $$
 By employing $W_p^{2,1}$ and $C^{\alpha,\,\alpha/2}(0<\alpha<1)$ estimates for parabolic
 equations (see \cite{LSU}), we derive that
 $$
 \|u_{\ep,\,n}\|_{W_p^{2.1}(\,\Omega_T^n \backslash B_\delta(P_0)\,)}
 +\|u_{\ep,\,n}\|_{C^{\alpha,\,\alpha/2}(\,\overline{\Omega_T^n }\,)}\leq C,
 $$
 where $C$ is a constant independent of $\ep$. Hence, there exists a
 $u_n\in W_p^{2.1}(\,\Omega_T^n \backslash B_\delta(P_0)\,)\cap C(\overline{\Omega_T^n})$
 and a subsequence of $\{u_{\ep,\,n}\}$, such that
 as $\ep\rightarrow 0^+$,
 $$
 u_{\ep,\,n}\rightharpoonup u_n\;\;
 \mbox{in}\;\;W_p^{2.1}(\,\Omega_T^n \backslash B_\delta(P_0)\,)\;\;
 \mbox{weakly}\;\;\;\mbox{and}\;\;\;u_{\ep,\,n}\rightarrow u_n\;\;
 \mbox{in}\;\;C(\overline{\Omega_T^n}).
 $$
By applying the method in \cite{Fr1} or \cite{Yang}, we can prove
that
 $u_n$ is the solution of (\ref{eq4.2}). And (\ref{eq4.8})-(\ref{eq4.10})
 are the consequences of (\ref{eq4.5})-(\ref{eq4.7}) as $\ep\rightarrow 0^+$.

 Finally, we prove the uniqueness of the solution. Suppose $u_n^1$ and
 $u_n^2$ are two $W_{p,\;loc}^{2,1}(\Omega_T^n)\cap C(\overline{\Omega_T^n})$
 solutions of (\ref{eq4.2}) and denote
 $$
 {\cal N}\stackrel{\Delta}{=}\{(x,t)\in \Omega_T^n
 :u_n^1(x,t)<u_n^2(x,t)\}.
 $$
 Suppose ${\cal N}$ is not empty, then in the domain ${\cal N}$,
 $$
 Ke^x\leq u_n^1(x,t)<u_n^2(x,t),\;\;
 \p_t u_n^2-{\cal L}u_n^2=c,\;\;
 \p_t (u_n^1-u_n^2)-{\cal L}(u_n^1-u_n^2)\geq0.
 $$
 Denote $W=u_n^1-u_n^2$, we have
 $$
 \p_t W-{\cal L}W\geq0\;\;\mbox{in }{\cal N},\qquad
 W=0\;\;\mbox{on }\p_p{\cal N}.
 $$
 From the A-B-P maximum principle (see \cite{Tao})), we have
 $W\geq0$ in ${\cal N}$, which contradicts the definition of ${\cal N}$.
\end{proof}

\noindent {\bf Proof of Theorem \ref{th4.1}: }Rewrite (\ref{eq4.2})
as follows:
 \bee
 \left\{
 \begin{array}{l}
 \p_t u_n-{\cal L}u_n=f(x,t),
 \hspace{3cm} (x,\,t)\in \Omega_T^n,
 \vspace{2mm} \\
 u_n(-n,\,\tau)\;={c\over r}+{rL-c\over r}\,e^{-r\tau},
 \qquad u_n(0,\,\tau)\;=K,\qquad
 0\leq \tau\leq T,
 \vspace{2mm} \\
 u_n(x,0)\;=\max\{L,Ke^x\},\hspace{1cm}\qquad\quad
 -n\leq x\leq 0.
 \end{array}
 \right.
 \eee
 Since $u_n\in W_{p,\,{\rm loc}}^{2,1}(\Omega_T^n)$, then we have
 $f(x,t)\in  L^p_{loc}(\Omega_T^n)$
 and
 $$
  f(x,t)=cI_{\{u_n>Ke^x\}}+qKe^xI_{\{u_n=Ke^x\}}.
 $$

 By the $W_p^{2,\,1}$ and $C^{\alpha,\,\alpha/2}$ estimates for parabolic
 equations (see \cite{LSU}), we deduce that for any fixed $R>\delta>0$, the following
 estimates hold
 \be                                                                                       \label{eq4.14}
 \|u_n\|_{W_p^{2,1}(\Omega_T^R \backslash B_\delta(P_0))} \leq C_{R,\delta},\qquad
 \|u_n\|_{C^{\alpha,\,\alpha/2}(\overline\Omega_T^R )} \leq C_R,
 \ee
 where $C_{R,\delta}$ depends on $R$ and $\delta,\,C_R$ depends on $R$, but they are independent of $n$.
 Then we derive that there exists a function
 $u\in W_{p,loc}^{2,\,1}(\,\Omega_T)\cap C(\overline{\Omega}_T)$
 and a function subsequence of $\{u_{n}\}$
 such that for any $R>\delta>0,\,p>1$,
 $$
 u_n\rightharpoonup u \;\;\mbox{in}\;\;W_p^{2.1}
 (\Omega_T^R \backslash B_\delta(P_0))\;\mbox{weakly}
 \quad\quad \mbox{as}\quad n\rightarrow+\infty.
 $$
 Moreover, (\ref{eq4.14}) and the imbedding theorem imply that
 \be                                                                                          \label{eq4.15}
 u_n\rightarrow u \;\;
 \mbox{in}\;\;C(\overline{\Omega}_T^R )\quad
 \mbox{and }\p_x u_n\rightarrow \p_x u \;\;
 \mbox{in}\;\;C(\overline{\Omega}_T^R \backslash
 B_\delta(P_0))\quad\mbox{as}\quad n\rightarrow+\infty.\quad
 \ee
 By the method in  \cite{Fr1} or \cite{Yang}, we can deduce that
 $u$ is the strong solution of (\ref{eq4.1}). Moreover, (\ref{eq4.15})
 implies that $\p_x u\in C(\Omega)$. And (\ref{eq4.11})-(\ref{eq4.13})
 are the consequences of (\ref{eq4.8})-(\ref{eq4.10}). The
 proof of the uniqueness is similar to the uniqueness proof in Lemma \ref{lm4.2}.
 \hfill$\Box$

\section{ The explicit solution of the PDE (\ref{eq2.5}).}

We present the explicit solution of the
 PDE (\ref{eq2.5}) in this appendix.
Since (\ref{eq2.5}) is a degenerate backward problem, we make the
 transformation (\ref{eq3.1}) as for the VI
 (\ref{eq2.10}).
Then it is not difficult to check that $u$ is governed by
 \be                                                                             \label{eq3.2}
 \left\{
 \begin{array}{l}
 \p_\tau u-{\cal L} u=c
 \hspace{3cm}\mbox{in}\;\; \Omega_T,
 \vspace{2mm} \\
 u(0,\,\tau)\;=K,\hspace{2.6cm}
 0\leq \tau\leq T,
 \vspace{2mm} \\
 u(x,0)\;=\max\{L,Ke^x\},\hspace{1cm}
 x\leq 0,
 \end{array}
 \right.
 \ee

It is standard to show that the classical solution of (\ref{eq3.2})
has the following integral
 expression (see for example \cite{Jiang}):
 \be\nonumber
 u(x,\tau)&\!\!\!=\!\!\!&Ke\,^x
 +c\,\int_0^\tau\,\Phi_1(-x,\tau,t)\,dt
 -q\,K\,e\,^x\,\int_0^\tau\,\Phi_2(-x,\tau,t)\,dt
 -c\,e^{-2\alpha_1x}\int_0^\tau\,\Phi_1(x,\tau,t)\,dt
 \\[2mm]\nonumber
 &\!\!\!\!\!\!&+\;q\,K\,e^{-2\alpha_1x-x}\int_0^\tau\,\Phi_2(x,\tau,t)\,dt
 +L\,\Phi_1(\ln L-\ln K-x,\tau,0)
 -K\,e\,^x\,\Phi_2(\ln L-\ln K-x,\tau,0)
 \\[2mm]                                                                             \label{eq3.5}
 &\!\!\!\!\!\!&-\,L\,e^{-2\alpha_1x}\,\Phi_1(\ln L-\ln K+x,\tau,0)
 +K\,e^{-2\alpha_1x-x}\,\Phi_2(\ln L-\ln K+x,\tau,0),
 \ee
 where
 \bee
 &&\Phi_1(x,\tau,t)=e^{r\,\!t-r\tau}\,\Phi(d_1(x,\tau,t)),\qquad\;\,
 \Phi_2(x,\tau,t)=e^{q\,\!t-q\tau}\,\Phi(d_2(x,\tau,t)),
 \\[2mm]
 &&{\displaystyle d_1(x,\tau,t)={x\over\sigma\sqrt{\tau-t}}
 -\sigma\,\alpha_1\,\sqrt{\tau-t},\quad
 d_2(x,\tau,t)={x\over\sigma\sqrt{\tau-t}}
 -\sigma\,(\alpha_1+1)\,\sqrt{\tau-t},}\qquad
 \\[2mm]
 &&\Phi(x)={\displaystyle {1\over\sqrt{2\pi}}
 \,\int_{-\infty}^x\,e^{-y^2/2}\,dy, \hspace{1.7cm}
 \alpha_1=-{1\over 2}+{r-q\over \sigma^2}.}
 \eee

 \end{document}